\documentclass[a4paper,UKenglish, autoref,numberwithinsect,cleveref]{lipics-v2021}
\nolinenumbers
\hideLIPIcs
\usepackage[utf8]{inputenc} 
\usepackage[ruled,vlined,algo2e,noend,linesnumbered]{algorithm2e}
\usepackage{dsfont}
\usepackage{amsfonts}
\usepackage{booktabs}

\newtheorem{rrule}[theorem]{Reduction Rule}

\crefname{rrule}{Reduction Rule}{Reduction Rules}

\newcommand{\Oh}{\ensuremath{\mathcal{O}}}

\newcommand{\bw}{bw-threshold dominating set}
\newcommand{\BW}{\textsc{BW-Threshold Dominating Set}}

\newcommand{\ramsey}{R_c}
\newcommand{\ramseyb}{Q_c}

\newcommand{\problemdef}[3]{
  \begin{center}
    \begin{minipage}{0.95\textwidth}
      \normalsize\textsc{#1} \smallskip \\
      \begin{tabularx}{\textwidth}{@{}l@{\hspace{3pt}}X}
        \normalsize\textbf{Input:}    & \normalsize#2 \\
        \normalsize\textbf{Question:} & \normalsize#3
      \end{tabularx}
    \end{minipage}
  \end{center}
}

\title{Exploiting \texorpdfstring{$c$}{c}-Closure in Kernelization Algorithms for Graph Problems}
\titlerunning{Exploiting \texorpdfstring{$c$}{c}-Closure in Kernelization Algorithms for Graph Problems}

\author{Tomohiro Koana}{Technische Universität Berlin, Algorithmics and Computational Complexity, Germany}{tomohiro.koana@tu-berlin.de}{https://orcid.org/0000-0002-8684-0611}{Supported by the Deutsche Forschungsgemeinschaft (DFG), project FPTinP, NI 369/19.}

\author{Christian Komusiewicz}{Philipps-Universität Marburg, Fachbereich Mathematik und Informatik,  Marburg, Germany}{komusiewicz@informatik.uni-marburg.de}{https://orcid.org/0000-0003-0829-7032}{}

\author{Frank Sommer}{Philipps-Universität Marburg, Fachbereich Mathematik und Informatik,  Marburg, Germany}{fsommer@informatik.uni-marburg.de}{https://orcid.org/0000-0003-4034-525X}{Supported by the Deutsche Forschungsgemeinschaft (DFG), project MAGZ, KO~3669/4-1.}

\authorrunning{Tomohiro~Koana, Christian~Komusiewicz and Frank~Sommer}
 
\Copyright{Tomohiro~Koana, Christian~Komusiewicz and Frank~Sommer}

\EventEditors{Fabrizio Grandoni, Peter Sanders, and Grzegorz Herman}
\EventNoEds{3}
\EventLongTitle{28th Annual European Symposium on Algorithms (ESA 2020)}
\EventShortTitle{ESA 2020}
\EventAcronym{ESA}
\EventYear{2020}
\EventDate{September 7--9, 2020}
\EventLocation{Pisa, Italy (Virtual Conference)}
\EventLogo{}
\SeriesVolume{173}
\ArticleNo{29}




\acknowledgements{This work was started at the research retreat of the TU Berlin Algorithms and Computational Complexity group held in September 2019 at Schloss Neuhausen (Prignitz). We would like to thank the anonymous reviewers for their helpful comments.}

\begin{document}

\maketitle
\begin{abstract}
 A graph is~$c$-closed if every pair of vertices with at least~$c$ common neighbors is adjacent. The $c$-closure of a graph~$G$ is the smallest number~$c$ such that~$G$ is~$c$-closed. Fox et al.~[SIAM~J.~Comput.~'20] defined~$c$-closure and investigated it in the context of clique enumeration. We show that~$c$-closure can be applied in kernelization algorithms for several classic graph problems. We show that \textsc{Dominating Set} admits a kernel of size~$k^{\Oh(c)}$, that \textsc{Induced Matching} admits a kernel with~$\Oh(c^7 k^{8})$ vertices, and that \textsc{Irredundant Set} admits a kernel with~$\Oh(c^{5/2} k^3)$ vertices.  Our kernelizations exploit the fact that $c$-closed graphs have polynomially-bounded Ramsey numbers, as we show.
\end{abstract}

\section{Introduction}
The area of Parameterized Complexity~\cite{CFK+15,DF13} aims at understanding which properties of input data can be used in the design of efficient algorithms for problems that are hard in
general. These input properties are encapsulated in the notion of a
\emph{parameter}, a numerical value that can be attributed to each input instance~$I$. For
a given hard problem and parameter~$k$, the first aim is to find a \emph{fixed-parameter algorithm}, an algorithm that solves the problem
in $f(k)\cdot |I|^{\Oh(1)}$~time. Such an algorithm is efficient when~$f$ grows moderately and $k$~takes on small
values. A second aim is to provide a \emph{kernelization}. This is an algorithm that given any instance~$(I,k)$ of a parameterized problem computes in polynomial time an equivalent instance of size~$g(k)$. If~$g$ grows not too much and~$k$ takes on small values, then a kernelization provably shrinks large input instances and thus gives a guarantee for the efficacy of data reduction rules. 
A central part of the design of good parameterized algorithms is thus the identification
of suitable parameters.
\begin{table}[t]
  {\footnotesize\caption{A comparison of the $c$-closure with the number~$n$ of vertices, number~$m$ of edges, and the maximum degree~$\Delta$ in social and biological networks.}}
  \label{tab:c-closure}
  \centering
  \begin{tabular}{lrrrr}
    \toprule
    Instance name & $n$& $m$ & $\Delta$ & $c$\\
    \midrule
    adjnoun-adjacency & 
112 & 425 & 49 & 14 \\ 
arenas-jazz &
198 & 2\,742 & 100 & 42\\ 
ca-netscience & 
379 & 914 & 34 & 5\\ 
bio-celegans &
453 & 2\,025 & 237 & 26 \\
bio-diseasome &
516 & 1\,188 & 50 & 9 \\
soc-wiki-Vote &
889 & 2\,914 & 102 & 18 \\ 
arenas-email &
1\,133 & 5\,451 & 71 & 19 \\
bio-yeast  &
1\,458 & 1\,948 & 56 & 8 \\
ca-CSphd  &
1\,882 & 1\,740 & 46 & 3\\ 
soc-hamsterster  &
2\,426 & 16\,630 & 273 & 77 \\
ca-GrQc  &
4\,158 & 13\,422 & 81 & 43 \\
soc-advogato  &
5\,167 & 39\,432 & 807 & 218 \\
bio-dmela  &
7\,393 & 25\,569 & 190 & 72 \\
ca-HepPh  &
11\,204 & 117\,619 & 491 & 90\\
ca-AstroPh  &
17\,903 & 196\,972 & 504 & 61\\ 
soc-brightkite & 
56\,739 & 212\,945 & 1\,134 & 184\\
    \bottomrule
  \end{tabular}
\end{table}

A good parameter should have the following advantageous
traits. Ideally, it should be easy to understand and compute.\footnote{This cannot always be
  guaranteed. For example, the important parameter treewidth is hard
  to compute and not as easily understood as simpler parameters.} It should take on small
values in real-world input data. It should describe input properties that are not captured
by other parameters. Finally, many problems should be amenable to parameterization using
this parameter. In other words, the parameter should help when designing fixed-parameter
algorithms or kernelizations.  

Fox et al.~\cite{FRSWW20} recently introduced the graph parameter~\emph{$c$-closure} which
describes a structural feature of many real-world graphs: When two vertices have many
common neighbors, it is likely that they are adjacent. More precisely, the $c$-closure of
a graph is defined as follows.
\begin{definition}[\cite{FRSWW20}]
  \label{defi-c-close-weakly-gamma}  
  A graph $G=(V,E)$ is \emph{$c$-closed} if every pair of vertices~$u\in V$ and~$v\in V$ with at
  least~$c$ common neighbors is adjacent. The $c$-closure of a graph is the smallest number~$c$ such that~$G$ is $c$-closed. 
\end{definition}
The parameter has many of the desirable traits mentioned above: it is easy to understand
and easy to compute, in $\Oh(n^\omega)$~time by squaring the adjacency matrix~\cite{FRSWW20}. 
Moreover, social networks are $c$-closed for relatively small values
of~$c$~\cite{FRSWW20}; refer also to \Cref{tab:c-closure}, where we provide the closure number of some further social networks.
In addition, the $c$-closure of
a graph gives a new class of graphs which is not captured by other popular measures. For example, every complete graph is 1-closed. Hence, a graph can have bounded
$c$-closure but unbounded degeneracy and thus unbounded treewidth. Conversely, the graph consisting of two vertices~$u$ and~$v$ and many vertex-disjoint~$u$-$v$-paths of length 2 is 2-degenerate, has treewidth 2, and unbounded~$c$-closure. Generally, one may observe that~$c$-closure is different from many common parameterizations which measure, in various ways, the sparseness of the input graph.  In this sense, the
structure described by the~$c$-closure of graphs is novel. The aim of this work is to show that~$c$-closure also has the final, most important trait: it helps when designing
fixed-parameter algorithms.

Fox et al.~\cite{FRSWW20} applied $c$-closure to the enumeration of maximal cliques, showing that a $c$-closed graph may have at most $3^{(c-1)/3}\cdot n^2$ maximal cliques. 
In combination with known clique enumeration algorithms, this implies that all maximal cliques of a graph can be enumerated in~$\Oh^*(3^{c/3})$ time.\footnote{The~$\Oh^*$ notation hides polynomial factors in the input size.}

How parameterization by $c$-closure helps, can also be seen rather easily for the
\textsc{Independent Set} problem. In \textsc{Independent Set} we are given an
undirected graph~$G=(V,E)$ and an integer~$k$ and want to determine whether~$G$ contains a
set of~$k$ vertices that are pairwise nonadjacent.  \textsc{Independent Set} is W[1]-hard when parameterized
by~$k$~\cite{CFK+15,DF13}. When one uses the maximum degree of~$G$ as an additional
parameter, then \textsc{Independent Set} has a trivial  kernelization: Any graph with at least~$(\Delta+1)k$ vertices has an independent set of size at least~$k$. With the following data reduction rule, we can obtain a kernelization for the combination of~$c$~and~$k$.
\begin{rrule}
  \label{rr:is}
  If~$G$ contains a vertex~$v$ of degree at least~$(c-1)(k-1)+1$, then remove~$v$ from~$G$.
\end{rrule}
To see that \Cref{rr:is} is correct, we need to show that the resulting graph~$G'$ has an independent set~$I$ of size~$k$ if and only if the original graph~$G$ has one. 
The nontrivial direction to show is that if~$G$ has an independent set~$I$ of size~$k$, then so does~$G'$.
Since this clearly holds for~$v \notin I$, we assume that~$v \in I$. 
To replace~$v$ by some other vertex, we make use of the~$c$-closure:
Every vertex~$u$ in~$I\setminus \{v\}$ has at most~$c-1$ neighbors in common with~$v$ since~$u$ and~$v$ are nonadjacent.
Thus, at most~$(c-1)(k-1)$ neighbors of~$v$ are also neighbors of some vertex in~$I\setminus \{v\}$.
Consequently, some neighbor~$w$ of~$v$ has no neighbors in~$I\setminus \{v\}$ and, therefore,~$(I\setminus \{v\})\cup \{w\}$ is an independent set of size~$k$ in~$G'$.

Applying \Cref{rr:is} exhaustively results in an instance with maximum degree less than~$ck$ which, due to the discussion above, directly gives the following. 
\begin{proposition}
  \textsc{Independent Set} admits a kernel with at most~$ck^2$ vertices.
\end{proposition}
Motivated by this simple result for a famous graph problem, we study how $c$-closure can
be useful for further classic graph problems when they are parameterized by a combination
of~$c$ and the solution size parameter~$k$. We obtain the following positive results. In
\Cref{sec:ds}, we show that \textsc{Dominating Set} admits a kernel
of size~$k^{\Oh(c)}$ and show that this kernelization is asymptotically optimal with respect to the dependence of the exponent on~$c$. Our results also hold for the more general \textsc{Threshold Dominating Set} problem where each vertex needs to be dominated~$r$ times.
In \Cref{sec:im}, we show that \textsc{Induced Matching} admits a kernel
with~$\Oh(c^{7} k^{8})$~vertices by means of LP relaxation of \textsc{Vertex Cover}.
Finally in \Cref{sec:irs}, we show that \textsc{Irredundant Set} admits a kernel
with~$\Oh(c^{5/2} k^3)$~vertices.
All kernelizations exploit a bound on Ramsey numbers for~$c$-closed graphs, which we prove in \Cref{sec:ramsey}.
This bound is---in contrast to Ramsey numbers of general graphs---polynomial in the size of a sought clique and independent set. We believe that this bound on the Ramsey numbers is of independent interest and that it provides a useful tool in the design of fixed-parameter algorithms for further problems on~$c$-closed graphs.

\paragraph*{Further Related Work} Very recently, there have been further studies of $c$-closure and weak $\gamma$-closure, a closely related parameter.\footnote{Informally,  weak~$\gamma$-closure relates  to $c$-closure in the same way that degeneracy relates to the maximum degree. The weak $\gamma$-closure of a graph is simultaneously upper-bounded by its $c$-closure and by its degeneracy.} 
Lokshtanov and Surianarayanan~\cite{LS21} showed that \textsc{Dominating Set} admits an FPT-algorithm for the smaller parameter~$k+\gamma$.
It was shown, for example, that various problems related to finding dense subgraphs are fixed-parameter tractable with respect to~$c$ or~$\gamma$~\cite{HR20,KKS20a} and that certain hard variants of \textsc{Vertex Cover}, such as \textsc{Capacitated Vertex Cover}, admit kernels of size~$k^{\Oh(\gamma)}$, where~$k$ is the solution size~\cite{KKS21}.
Furthermore, we studied a generic local graph partitioning problem generalizing \textsc{Densest~$k$-Subgraph}, \textsc{Partial Vertex Cover}, and \textsc{Max~$(k,n-k)$-Cut}~\cite{KKNS22}. 
Kanesh et al.~\cite{KMR+22} studied further domination problems like \textsc{Connected Dominating Set}, \textsc{Partial Dominating Set}, and \textsc{Perfect Code} in $c$-closed graphs.
Finally, it has been shown for some subgraphs of constant size that they can be detected efficiently in graphs with small~$c$-closure~\cite{KN21}.



\section{Preliminaries}

For $m \le n \in \mathds{N}$, we write $[m, n]$ for the set $\{ m, m + 1, \dots, n \}$ and $[n]$ for $[1, n]$.
For a graph~$G$, we denote its vertex set and edge set by $V(G)$ and~$E(G)$, respectively.
Let~$X, Y \subseteq V(G)$ be vertex subsets.
We use $G[X]$ to denote the subgraph induced by~$X$.
We also use~$G[X, Y]:=(X \cup Y, \{ xy \in E(G) \mid x \in X, y \in Y \})$ to denote the bipartite subgraph induced by disjoint vertex sets $X$ and~$Y$.
We let~$G - X$ denote the graph obtained by removing vertices in $X$.
We denote by~$N_G(X):=\{ y \in V(G) \setminus X \mid xy \in E(G), x \in X \}$ and~$N_G[X]:=N_G(X) \cup X$, the open and closed neighborhood of $X$, respectively.
For all these notations, when $X$ is a singleton~$\{ x \}$ we may write~$x$ instead of~$\{x\}$.
We drop the subscript~$\cdot_G$ when it is clear from context.
Let~$v \in V(G)$.
We denote the degree of $v$ by $\deg_G(v)$.
We call~$v$ \emph{isolated} if $\deg_G(v) = 0$ and \emph{non-isolated} otherwise.
We also say that $v$ is a \emph{leaf vertex} if~$\deg_G(v) = 1$ and a \emph{non-leaf vertex} if~$\deg_G(v) \ge 2$.
Moreover, we say that $v$ is \emph{simplicial} if~$N_G(v)$ is a clique.
The maximum and minimum degree of~$G$ are~$\Delta_G := \max_{v \in V(G)} \deg_G(v)$ and~$\delta_G := \min_{v \in V(G)} \deg_G(v)$, respectively.
The degeneracy of~$G$ is $d_G := \max_{S \subseteq V(G)} \delta_{G[S]}$.
A graph~$G$ has girth~$g$ if the shortest cycle in~$G$ has length~$g$.

In this paper, we investigate the parameterized complexity of various problems whose input comprises a graph~$G$ and an integer~$k$.
A problem is \emph{fixed-parameter tractable} if it can be solved in $f(k) \cdot n^{\Oh(1)}$ time where $n := |V(G)|$ and $f$~is some computable function.
Two instances $(G, k)$ and $(G', k')$ are \emph{equivalent} if $(G, k)$ is a Yes-instance if and only if $(G', k')$ is a Yes-instance.
A \emph{kernelization algorithm} is a polynomial-time algorithm which transforms an instance~$(G, k)$ into an equivalent instance $(G', k')$ such that $|V(G')| + k' \le g(k)$, where $g$ is some computable function.
It is well-known that a problem is fixed-parameter tractable if and only if it admits a kernelization algorithm. 

Our kernelization algorithms consist of a sequence of \emph{reduction rules}.
Given an instance~$(G, k)$, a reduction rule computes an instance~$(G', k')$.
We will develop kernelization algorithms for $c$-closed graphs.
For our purposes, we say that a reduction rule is \emph{correct} if the input instance $(G, k)$ for a $c$-closed graph $G$ is equivalent to the resulting instance $(G', k')$ and $G'$~is also $c$-closed. 
For more information on parameterized complexity, we refer to the standard monographs~\cite{CFK+15,DF13}.

We will make use of the following observations throughout this work.
\begin{observation}
  \label{obs:removev}
  If~$G$ is~$c$-closed, then~$G - v$ is also~$c$-closed for any~$v \in V(G)$.
\end{observation}

\begin{observation}
  \label{obs:cliqueintersection}
  Let~$C$ be a maximal clique in a~$c$-closed graph~$G$.
  Then~$|C \cap N(v)| < c$ for every~$v \in V(G) \setminus C$.
\end{observation}

\begin{observation}
  \label{obs:addv}
  Let~$G$ be a~$c$-closed graph and let~$C$ be a
  \begin{itemize}
  \item clique of size at most~$c - 1$ in~$G$ or
  \item a maximal clique in~$G$.
  \end{itemize}
  Then, the graph~$G'$ obtained by attaching a simplicial vertex~$v$ to~$C$ (that is,~$N_{G'}(v) = C$) is~$c$-closed.
\end{observation}

\section{On Ramsey Numbers of \texorpdfstring{$\mathbf{c}$}{c}-Closed Graphs}
\label{sec:ramsey}
Ramsey's theorem states that there is a function~$R$ such that any graph~$G$ with at least~$R(q, b)$ vertices contains a clique of size~$q$ or an independent set of size~$b$, for any~$q, b \in \mathds{N}$. 
The numbers~$R(q, b)$ are referred to as Ramsey numbers. 
It is known that~$R(t, t) > 2^{t / 2}$ for any~$t \ge 3$~\cite{Erd47,Juk11} and hence~$R(t, t)$ grows exponentially with~$t$. 
Here, we show that the Ramsey number~$R(q,b)$ is actually polynomial in~$q$ and~$b$ in~$c$-closed graphs.

Our proof relies on the notion of \emph{2-maximal independent sets}:
We say that a maximal independent set $I$ is \emph{$2$-maximal} in $G$, if it holds for every $x, y\in V(G) \setminus I$ and~$v \in I$ that $(I \setminus \{ v \}) \cup \{ x, y \}$ is not an independent set.
First, we prove a property of 2-maximal independent sets that plays an important role in the proof of the Ramsey bound in $c$-closed graphs.

\begin{lemma}
  \label{lemma:privateclique}
  Let $I$ be a 2-maximal independent set and for every $x \in I$, let $I_x := \{ v \in V(G) \mid N_G[v] \cap I = \{ x \} \}$ be the set of vertices adjacent to $x$ but nonadjacent to all other vertices in $I$.
  Then, $I_x$ is a clique.
\end{lemma}
\begin{proof}
  Assume to the contrary that there exist two vertices $u, v \in I_x$ such that $uv \notin E(G)$.
  Then, $(I \setminus \{ x \}) \cup \{ u, v \}$ is also an independent set, a contradiction to the fact that $I$ is 2-maximal.
\end{proof}

Let~$\ramsey(q, b) := (c - 1) \cdot \binom{b - 1}{2} + (q - 1)(b - 1) + 1$.
We show that $\ramsey(q, b)$ is an upper bound on the Ramsey number in $c$-closed graphs.

\begin{lemma}
  \label{lemma:ramsey}
  Let $G$ be a~$c$-closed graph~$G$ on at least~$\ramsey(q, b)$ vertices.
  Then, $G$~contains a clique of size~$q$ or an independent set of size~$b$.
  Moreover, there is an algorithm that finds a clique of size $q$ or an independent set of size $b$ in polynomial time.
\end{lemma}
\begin{proof}
  Assume to the contrary that~$G$ has no clique of size~$q$ and no independent set of size~$b$. 
  Let $I$ be a 2-maximal independent set of $G$. 
  For every~$x \in I$, let~$I_x := \{ v \in V(G) \mid N_G[v] \cap I = \{ x \} \}$ be the set of vertices adjacent to $x$ but nonadjacent to all other vertices in~$I$.
  By \Cref{lemma:privateclique}, $I_x$ is a clique for every $x \in I$ and thus $|I_x| < q$.
  Since $I$ is inclusion-wise maximal, every vertex in $G$ is adjacent to at least one vertex of $I$.
  It follows that
  \begin{align*}
    |V(G)| \le \sum_{x \in [I]} |I_x| + \sum_{x \ne y \in [I]} |N(x) \cap N(y)|.
  \end{align*}
  Note that~$|I_x| \le q - 1$ for each~$x \in I$ and~$|N(x) \cap N(y)| \le c - 1$ for~$x \ne y \in I$ by the~$c$-closure of~$G$.
  Since~$|I|<b$, we have a contradiction on $|V(G)|$.

  Note also that the proof of \Cref{lemma:ramsey} is constructive:
  First, we find a 2-maximal independent set $I$.
  If $I$ has fewer than $b$ vertices, then there is a vertex~$x \in I$ with~$|I_x| \ge q$.
  Thus, one can find in polynomial time a clique of size $q$ or an independent set of size $b$ in $G$, provided that $G$ is a $c$-closed graph on more than~$(c - 1) \cdot \binom{b}{2} + (q - 1)(b - 1)$ vertices.
\end{proof}

The bound in \Cref{lemma:ramsey} is essentially tight:
Consider a graph $G$ consisting of a disjoint union of~$b - 1$ complete graphs, each of order~$q - 1$.
Note that $G$ is $c$-closed for any $c \in \mathds{N}$ and that~$G$ has no clique of size $q$ or independent set of size $b$.
Thus, we have a tight bound for $c = 1$.
This example also suggests that the bound in \Cref{lemma:ramsey} cannot be asymptotically improved for $q \ge cb$.

\section{(Threshold) Dominating Set}
\label{sec:ds}
In this section we show that \textsc{Threshold Dominating Set} admits a kernel with $k^{\Oh(cr)}$ vertices.
The problem is defined as follows.

\problemdef{Threshold Dominating Set}
{A graph~$G$ and $r, k\in\mathds{N}$.}
{Is there a vertex set~$D\subseteq V(G)$ such that~$|D|\le k$ and each vertex~$v\in V(G)$ is dominated by~$D$ at least~$r$ times, that is, $|N[v]\cap D|\geq r$?}

\textsc{Dominating Set} is the special case of \textsc{Threshold Dominating Set} when~$r=1$. 
\textsc{Dominating Set} is W[2]-hard when parameterized by~$k$ even in bipartite or split graphs~\cite{RS08}. Furthermore, \textsc{Dominating Set} was shown to remain NP-hard on graphs with girth at least~$t$ for any constant~$t$~\cite{AKL04}.
Hence, \textsc{Dominating Set} is NP-hard even on~$2$-closed graphs.

There are several fixed-parameter tractability results in restricted graph classes: 
When the graph~$G$ contains no induced~$C_3$ and no induced~$C_4$, \textsc{Dominating Set} admits a kernel with~$\Oh(k^3)$ vertices and \textsc{Threshold Dominating Set} is fixed-parameter tractable~\cite{RS08}.
Furthermore, \textsc{Dominating Set} in~$d$-degenerate graphs can be solved in~$k^{\Oh(dk)}n$~time~\cite{AG09}. 
This result was extended to an algorithm with running time~$k^{\Oh(dkr)}$ for \textsc{Threshold Dominating Set} in~$d$-degenerate graphs~\cite{GV08}. 
When the graph~$G$ does not contain the complete bipartite graph~$K_{i,j}$ for fixed~$j\le i$ as a (not necessarily induced) subgraph, \textsc{Dominating Set} admits a kernel of~$\Oh((j+1)^{i+1}k^{i^2})$ vertices which can be computed in~$\Oh(n^i)$~time~\cite{PRS12}. 
Since~$d$-degenerate graphs do not contain a~$K_{d+1,d+1}$ as a subgraph, \textsc{Dominating Set} admits a kernelization of~$\Oh(k^{(d+1)^2})$ vertices computable in polynomial time~\cite{PRS12}. 
This kernel size is essentially optimal since \textsc{Dominating Set} in~$d$-degenerate graphs admits no kernel of size~$\Oh(k^{(d-3)(d-1)-\epsilon})$ for any~$\epsilon>0$ unless NP~$\subseteq$ coNP/poly~\cite{CGH17}. 
When~$G$ does not contain the complete bipartite graph~$K_{t,t}$ as a (not necessarily induced) subgraph, \textsc{Dominating Set} can be solved in~$2^{\Oh(tk^2(4k)^t)}$~time~\cite{TV12}. Since each~$d$-degenerate graph does not contain a~$K_{d+1,d+1}$ as a subgraph, this extends the FPT-algorithm for~$k$ in~$d$-degenerate graphs~\cite{AG09}.
None of the above kernelizations and fixed-parameter algorithms implies a tractability result on $c$-closed graphs, since the respective structural restrictions on~$G$ all exclude cliques of some size. Moreover, since any graph without induced~$C_3$ and without induced~$C_4$ is 2-closed, our results extend the kernelization algorithms for these graphs to a more general class of graphs.

To obtain a kernel  for \textsc{Threshold Dominating Set} in~$c$-closed graphs we first provide a kernelization for a more general, colored variant defined as follows. 
The input graph is a \emph{bw-graph}, where the vertex set~$V(G)$ is partitioned into black vertices~$B$ and white vertices~$W$. 
We only require to dominate black vertices~$r$ times. 
The problem is defined as follows.

\problemdef{\BW}
{A bw-graph~$G$ and~$r,k \in \mathds{N}$.}
{Does~$G$ contain a \emph{\bw} $D \subseteq V(G)$, that is, a set such that~$|N[v]\cap D|\ge r$ for each vertex~$v\in B$, of size at most~$k$?}

Clearly, each instance~$(G,k)$ of \textsc{Threshold Dominating Set} is equivalent to the instance~$(G,k)$ of \BW{} where each vertex is black.


\subsection{A Polynomial Kernel in \texorpdfstring{$c$}{c}-closed Graphs}
\label{sec-kernel-threshold-ds}

We first develop a kernelization algorithm for \BW{} and then we will remove colors at the end.
We start with an auxiliary lemma which will simplify subsequent proofs in this section.

\begin{lemma}
  \label{lemma-simplical-vertex}
  Let~$(G,k)$ be a Yes-instance of \BW{} and let~$v$ be a simplicial vertex with at least~$r$ neighbors. Then, there exists a \bw~$D$ of size at most~$k$ such that $v \notin D$.
\end{lemma}
\begin{proof}
  Suppose that~$G$ has a \bw~$D$ of size at most~$k$.
  We are immediately done if~$v \notin D$, so we can assume that~$v \in D$.
  If~$N[v] \subseteq D$, then~$D \setminus \{ v \}$ is a \bw{} of size at most~$k$.
  Otherwise, there is a vertex $u \in N(v) \setminus D$ and $(D \setminus \{ v \}) \cup \{ u \}$ is a \bw{} of size at most~$k$ not containing~$v$. 
\end{proof}

Our algorithm is outlined as follows.
We begin with \Cref{lemma:findclique}, which finds some clique~$C$ of one of two types, whenever the graph has a vertex $v$ with $\rho := \ramsey(ck, ck^{c - 1} + 1)$ black neighbors. The first type of cliques are maximal cliques with many black vertices. The second type of cliques are small cliques with a certain large independent set in the common neighborhood.  
We then apply reduction rules to reduce the number of black neighbors of $v$, taking advantage of structural properties of such cliques in $c$-closed graphs.
More specifically, we use \Cref{rr:clique,rr:cliquei} when $C$ is of the first and second type, respectively.
We thus end up with an instance in which every vertex has at most $\rho$ black neighbors.
If there still remain more than $k \rho$ black vertices, then the instance is a No-instance (\Cref{rr-number-blue-vertices}).
Finally, we obtain an upper bound on the number of white vertices in \Cref{lemma-bound-number-white-vertices} as well.

\begin{lemma}
  \label{lemma:findclique}
  Suppose that there is a vertex $v$ with at least $\rho = \ramsey(ck, ck^{c - 1} + 1)$ black neighbors.
  Then, we can find in polynomial time a clique $C$ that contains $v$ and is of one of the following types:
  \begin{enumerate}
    \item[(1)]
      $C$ is maximal and contains at least $ck$ black vertices.
    \item[(2)]
      $C$ is of size exactly $c - i$ for some $i \in [c - 1]$ and there is an independent set $I \subseteq \bigcap_{v' \in C} N(v')$ of $c k^i + 1$ black vertices such that every vertex outside $C$  has at most $c k^{i - 1}$ neighbors in~$I$.
  \end{enumerate}
\end{lemma}
\begin{proof}
  Recall that given any $c$-closed graph on at least $\ramsey(ck, c k^{c - 1} + 1)$ vertices, we can find in polynomial time a clique of size $ck$ or an independent set of size~$c k^{c - 1}+1$ by \Cref{lemma:ramsey}.
  By applying \Cref{lemma:ramsey} on $G[N(v) \cap B]$, we can thus find (a) a clique~$C'$ of~$ck$ black vertices containing~$v$ or (b) an independent set $I'$ of $c k^{c - 1} + 1$ black vertices in $N(v)$ in polynomial time.

  For case (a), consider the set $C'' := \{ v'' \in V(G) \mid N(v'') \supseteq C' \}$ of vertices adjacent to all vertices of $C'$.
  Since any pair of vertices in $C''$ has at least $|C'| \ge c$ neighbors, $C''$ forms a clique.
  We thus have that $C := C' \cup C''$ is a clique.
  The maximality of $C$ follows from the definition of $C''$, showing that $C$ is of type (1).

  For case (b), we use the following polynomial-time procedure:
  Initially, let $C := \{ v \}$,~$I := I'$, and $i := c - 1$.
  As long as there is a vertex $u$ which has more than $c k^{i - 1}$ neighbors in~$I$, add $u$ to $C$, delete vertices nonadjacent to $u$ from $I$, and decrease~$i$ by~1.
  At the latest, this procedure halts when $i = c - |C| \ge 1$ because then the vertices in $I'$ have at most $c - 1$ common neighbors.
  We show that the resulting vertex set~$C$ is of type (2).
  By construction,~$I \subseteq \bigcap_{v' \in C} N(v')$ is an independent set of more than $c k^i$ black vertices.
  Moreover, there is no vertex outside $C$ with more than~$c k^{i - 1}$ neighbors in $I$.
  Since the vertices in $C$ have at least $|I| \ge c$ common neighbors, they are pairwise adjacent.
  We thus have shown that $C$ is a clique of type (2).
\end{proof}

We first present a reduction rule for the case that \Cref{lemma:findclique} finds a clique of type (1).

\begin{rrule}
  \label{rr:clique}
  Let~$C$ be a maximal clique containing at least~$ck$ black vertices.
  Then,
  \begin{enumerate}
    \item
      add a vertex~$u$ and add an edge~$uv$ for each~$v \in C$, 
    \item
      color~$u$ black, and
    \item
      color all the vertices in~$C$ white.
  \end{enumerate}
\end{rrule}

\begin{lemma}
  \label{lemma:rrclique}
  \Cref{rr:clique} is correct.
\end{lemma}
\begin{proof}
  Let~$D$ be a \bw{} of~$G$ of size at most~$k$. 
  We claim that~$|D \cap C| \ge r$. 
  Assume to the contrary that~$|D \cap C| \le r-1$. 
  By \Cref{obs:cliqueintersection}, each vertex in~$D \setminus C$ dominates at most~$c - 1$ vertices in~$C$. 
  Since~$C$ contains at least~$ck$ black vertices, there is a black vertex in~$C$ that is not dominated~$r$ times by~$D$, a contradiction. 
  Thus,~$|D\cap C|\ge r$. 
  Let~$G'$ be the graph obtained as a result of \Cref{rr:clique}.
  Since~$uv \in E(G)$ for each~$v\in C$, we see that~$|N_{G'}(u) \cap D|\ge r$ and thus~$D$ is also a \bw{} of the graph~$G'$.
  The other direction of the equivalence follows from the fact that~$u$ is simplicial and \Cref{lemma-simplical-vertex}. 
  Finally, note that \Cref{rr:clique} maintains the~$c$-closure by \Cref{obs:addv}.
\end{proof}

When \Cref{lemma:findclique} finds a clique of type (2), we apply the following reduction rule.

\begin{rrule}
  \label{rr:cliquei}
  Let $i \in [c - 1]$.
  Let~$C$ be a clique of size exactly~$c - i$ with an independent set $I \subseteq \bigcap_{v' \in I} N(v')$ of $c k^i + 1$ black vertices such that every vertex outside $C$ has at most $c k^{i - 1}$ neighbors in $I$.
  If $r \le c - i$, then
  \begin{enumerate}
    \item
      add a vertex~$u$ and add an edge~$uv$ for each~$v \in C$, 
    \item
      color~$u$ black, and 
    \item
      color all the vertices~in $C$ and~$I$ white.
  \end{enumerate}
  Otherwise, return No.
\end{rrule}

\begin{lemma}
  \label{lemma:cliquei}
  \Cref{rr:cliquei} is correct.
\end{lemma}
\begin{proof}
  We show that $|D \cap C| \ge r$ for any \bw~$D$ of size at most~$k$.
  Assume to the contrary that $|D \cap C| < r$.
  By assumption, we have that~$|I| > ck^i$ and that every vertex in $V(G) \setminus C$ dominates at most $ck^{i - 1}$ vertices of $I$.
  The pigeon-hole principle thus yields that there is a black vertex $x \in I$ not dominated by $D \setminus C$.
  It follows that~$D$ contains at least~$r$ vertices of~$C$.
  Consequently, for~$r > c - i=|C|$ the given instance is a No-instance.

  For~$r \le c - i$ let~$G'$ be the graph obtained as a result of \Cref{rr:cliquei}.
  Since~$|D \cap C| \ge r$ for any \bw{} of size at most $k$ in $G$, it is also a \bw{} in $G'$.
  The other direction follows from \Cref{lemma-simplical-vertex}. 
  Finally, note that $G'$ is~$c$-closed by \Cref{obs:addv}.
\end{proof}

As long as there is a vertex $v$ which has more than $\rho$ black neighbors, we invoke the algorithm of \Cref{lemma:findclique} to find a clique $C$ containing $v$.
We then apply \Cref{rr:clique} or \Cref{rr:cliquei}, depending on the type of $C$.
Observe that both reduction rules decrease the number of black neighbors of $v$ and hence we obtain in polynomial time an equivalent instance in which every vertex has at most $\rho$ black neighbors.
If the resulting instance has more than $k \rho$ black vertices, then it is clearly a No-instance:

\begin{rrule}
  \label{rr-number-blue-vertices}
  If~$G$ contains more than~$k \rho$ black vertices, then return~No.
\end{rrule}

To compute a kernel it  remains to upper-bound the number of white vertices in~$G$. 

\begin{rrule}
  \label{rr-remove-white-vertices}
  Let~$w$ be a white vertex in~$G$. 
  If there exist at least~$r$ further vertices~$v_1, \ldots , v_r$ such that~$N(w)\cap B\subseteq N[v_i]\cap B$ for each~$i\in[r]$, then remove~$w$.
\end{rrule}

It is easy to see that \Cref{rr-remove-white-vertices} can be applied exhaustively in polynomial time.

\begin{lemma}
\label{lemma-rr-remove-white-vertices-correct}
\Cref{rr-remove-white-vertices} is correct. 
\end{lemma}
\begin{proof}
  Let~$G':=G-w$. 
  Suppose that~$G$ has a \bw~$D$ of size at most~$k$.
  If~$w \notin D$, then~$D$ is also a \bw{} of~$G'$.
  Hence, we can assume that~$w \in D$.
  If~$v_i\in D$ for all~$i \in [r]$, then~$D \setminus \{ w \}$ is a \bw{} for~$G$ and hence also for~$G'$. 
  Otherwise, there exists some~$i\in[r]$ with~$v_i\notin D$. 
  Since~$N(w)\cap B\subseteq N[v_i]\cap B$, the set~$(D \setminus \{ w \}) \cup \{ v_i \}$ is a \bw{} of size at most~$k$ of~$G$ and $G'$.
  The other direction follows trivially. 
  Observe that removing vertices maintains the $c$-closure by \Cref{obs:removev}.
\end{proof}

In the following, we will assume that \Cref{rr-remove-white-vertices} has been applied exhaustively. 
Now, we obtain a bound on the number of white vertices in~$G$.

\begin{lemma}
\label{lemma-bound-number-white-vertices}
The graph~$G$ contains~$\Oh(c|B|^2+ |B|^{r-1})$ white vertices.
\end{lemma}

\begin{proof}
  Since \Cref{rr-remove-white-vertices} has been applied exhaustively,~$G$ contains at most~$r$~white vertices~$w$ such that~$N(w)\subseteq W$. 
  Hence, it remains to bound the number of white vertices with at least one black neighbor. 
  Observe that by the $c$-closure of $G$, there are $\Oh(c\, |B|^2)$ white vertices that are neighbors of two nonadjacent vertices~$u,v\in B$.

  For all remaining white vertices~$w$, the set $B_w := N(w) \cap B$ of black neighbors is a clique.
  Since \Cref{rr-remove-white-vertices} has been applied  exhaustively, we have $|B_w| < r$.
  Moreover, also because \Cref{rr-remove-white-vertices} has been applied, for each clique $C\subseteq B$ of size $i \in [r - 1]$, there are at most $r - i$ white vertices with $B_w = C$.
  Thus, the number of white vertices $w$ such that~$B_w$ is a clique is
  \begin{align*}
    \sum_{i = 1}^{r - 1} i \, |B|^{r - i} = \frac{|B|(|B|^r - 1)}{(|B| - 1)^2} - \frac{|B|}{|B| - 1}r \in \Oh(|B|^{r - 1}).
  \end{align*}
  Overall, there are $\Oh(c\,|B|^2+ |B|^{r-1})$ white vertices.
\end{proof}

Recall that there are $k\rho \in \Oh(c^3 k^{2c - 1})$ black vertices by \Cref{rr-number-blue-vertices}.
So the overall number of vertices is $\Oh(c^7 k^{4c-2} + c^{3(r-1)}k^{(2c - 1)(r - 1)})$, resulting in the following theorem:

\begin{theorem}
  \label{theo-kernel-rwb-threshold-ds}
  \BW{} has a kernel with~$k^{\Oh(cr)}$ vertices. 
\end{theorem}
To obtain a kernel for \textsc{Threshold Dominating Set}, it remains to show that any \BW{} instance can be transformed into an equivalent instance of \textsc{Threshold Dominating Set}.

\begin{theorem}
  \label{theo-kernel-threshold-ds}
  \textsc{Threshold Dominating Set} has a kernel with~$k^{\Oh(cr)}$ vertices.
\end{theorem}

\begin{proof}
  To obtain a~$k^{\Oh(cr)}$-vertex kernel for \textsc{Threshold Dominating Set}, we first construct an equivalent instance~$(G, k)$ of \BW{} using \Cref{theo-kernel-rwb-threshold-ds}.
  Then, we transform~$(G, k)$ into an equivalent instance~$(G',k')$ of \textsc{Threshold Dominating Set} in~$k^{\Oh(cr)}$-closed graphs as follows.

  We start with a copy of~$G$. 
  We add a clique~$Q:=\{ w_1, \dots, w_{r + 1} \}$ of~$r + 1$ vertices. 
  Then, for each white vertex~$w$ we add edges~$w w_1, \dots, w w_{r}$.
  Then, we remove all vertex colors.
  We call the resulting graph $G'$.
  Let~$C := \{ w_1, \ldots, w_r\}$ and let~$k' := k + r$.
  We show that~$(G, k)$ is a Yes-instance if and only if~$(G', k')$ is a Yes-instance.
  
  Let~$D$ be a \bw{} of~$G$.
  By construction, $D \cup C$ is a threshold dominating set of size at most~$k'$ of~$G'$.
  Conversely, suppose that~$G'$ has a threshold dominating set~$D'$ of size at most~$k'$. 
  By \Cref{lemma-simplical-vertex}, we can assume that~$w_{r+1}\notin D'$. 
  Since~$\deg_{G'}(w_{r + 1}) = r$, it holds that~$N_{G'}(w_{r+1}) = C \subseteq D'$. 
Since~$C$ dominates only vertices of~$G'$ that are white in~$G$, we have that every vertex of~$G'$ which is black in~$G$ is dominated~$r$ times by~$D:=D'\setminus C$.  
  Thus,~$D$ is a \bw{} of size~$k$ for~$G$. 
\end{proof}
Since the kernelization does not change the parameter~$r$, it also gives a kernelization for \textsc{Dominating Set}.

\begin{corollary}
  \label{cor:dskernel}
  \textsc{Dominating Set} has a kernel with~$k^{\Oh(c)}$ vertices.
\end{corollary}

To complement this result, we show that there is no kernel for \textsc{Dominating Set} significantly smaller than that of \Cref{cor:dskernel} under a widely believed assumption.

\begin{theorem}
  For $c \ge 3$, \textsc{Dominating Set} has no kernel of size~$\Oh(k^{c-1-\epsilon})$ unless \rm{coNP~$\subseteq$ NP/poly}.
\end{theorem}
\begin{proof}
We will show the theorem by a reduction from \textsc{$\lambda$-Hitting Set}. 

\problemdef
{$\lambda$-Hitting Set}
{A set family~$\mathcal{F}$ over a universe $U$, where each~$S\in\mathcal{F}$ has size~$\lambda$, and~$k\in\mathds{N}$.}
{Is there a subset~$X \subseteq U$ of size at most~$k$ such that for each~$S\in\mathcal{F}$ we have~$X \cap S\neq\emptyset$?}

For any~$\lambda \ge 2$, \textsc{$\lambda$-Hitting Set} does not have a kernel of size~$\Oh(k^{\lambda-\epsilon})$ unless coNP~$\subseteq$ NP/poly~\cite{DM12,DM14}. 
Let~$(U,\mathcal{F},k)$ be an instance of \textsc{$\lambda$-Hitting Set}. 
We will construct a~$(\lambda+1)$-closed graph~$G$ as follows: 
The vertex set~$V(G)$ is~$U \cup \mathcal{F}$.  
We add edges such that~$U$ forms a clique in~$G$. 
We also add an edge between~$u \in U$ and~$S \in \mathcal{F}$ if and only if~$u \in S$. 
Finally, we set~$k'=k$. 
Since~$\deg_G(S) = \lambda$ for each~$S\in\mathcal{F}$, the graph~$G$ is~$(\lambda+1)$-closed. 

By construction, each hitting set~$X$ of size at most~$k$ is also a dominating set of size at most~$k$ of~$G$.
  For the converse direction, we may assume  by \Cref{lemma-simplical-vertex} that there is a dominating set~$D$ of size at most~$k$ for~$G$ not containing any vertex from~$\mathcal{F}$.
  Thus,~$D$ is also a hitting set of $(U, \mathcal{F}, k)$.
  
  Observe that our reduction preserves the parameter (that is, $k = k'$). 
  Thus, it follows from the result of Hermelin and Wu~\cite{HW12} that if \textsc{Dominating Set} admits a kernel of size~$\Oh(k^{\lambda-1-\epsilon})$ for some~$\epsilon>0$, then \textsc{$\lambda$-Hitting Set} admits a kernel of size~$\Oh(k^{\lambda-\epsilon})$, implying that coNP~$\subseteq$ NP/poly~\cite{DM12,DM14}. 
\end{proof}

\subsection{A Faster Algorithm}
\label{sec-faster-fpt-threshold-ds}


\Cref{theo-kernel-threshold-ds} yields a simple FPT algorithm to solve an instance~$(G,k)$ of \textsc{Threshold Dominating Set}.
First, compute a kernel~$(G',k')$ of~$k^{\Oh(cr)}$ vertices and then 
run a brute-force algorithm on the kernel checking whether there exists a set of~$k'$ vertices that dominates all vertices of~$V(G')$.
The running time of this algorithm is~$\Oh^*\big(k^{\Oh(crk)}\big)$. 
In this section we will present a faster FPT algorithm \textit{SolveTDS} with running time~$\Oh^*\big((ck)^{\Oh(rk)}\big)$. 
The pseudocode is shown in Algorithm~\ref{algo-branching-threshold-ds}.

\begin{algorithm2e}[t]
	\caption{The FPT algorithm \textit{SolveTDS} to solve \textsc{Threshold Dominating Set}.}
	\label{algo-branching-threshold-ds}
\DontPrintSemicolon
        \SetKwFunction{algo}{\textit{SolveTDS}}
        \SetKwFunction{proc}{\textit{Branch}}
  \SetKwProg{myalg}{Algorithm}{}{}
  \myalg{\algo{$G,k$}}{
    Color each vertex black \label{line-transform-into-wb-tds}\\
   \proc{$G,k,\emptyset$}\label{line-initial-branching}\
}{}
  \SetKwProg{myproc}{Procedure}{}{}
  \myproc{\proc{$G=((B,W),E),k',D$}}{
  Color each black vertex which is dominated by~$D$ at least~$r$ times white  \label{line-update-demands}\\
  \lIf{$|B|\ge 1$ and~$k'=0$}{\KwRet No} \label{line-b-zero-k-one}
  \While{\texttt{true}\label{line-while-true}}{
    Compute greedily a~$2$-maximal independent set~$I$ in~$G[B]$ \label{line-compute-greedy-is}\\
    \eIf{$|I|\ge k'+1$\label{line-inner-branching}}{Let~$I'\subseteq I$ be an arbitrary set such that~$|I'|=k'+1$ \label{line-compute-set-size-k-plus-one}\\ Let~$P$ be the set of common neighbors of at least two vertices in~$I'$ \label{line-compute-common-neigs-of-i}\\\lForEach{$v \in P\setminus D$}{\proc{$G,k'-1, D\cup\{v\}$} \label{line-choose-v-in-p}}}
    {\label{line-else}\eIf{private black neighborhood~$B_v$ for some~$v\in I$ is a clique of size at least~$ck$\label{line-many-priviate-black-neighbors}}{Apply \Cref{rr:clique} on~$B_v\cup\{v\}$ \label{line-apply-rr-clique}}{Apply \Cref{rr-remove-white-vertices} exhaustively on~$(G,k')$ \label{line-apply-rr-remove-white}\\
    \KwRet the result of brute-force search on~$(G,k')$ \label{line-brute-force}\
    }}}
  %
  }
\end{algorithm2e}

As in the kernelization algorithm presented in \Cref{sec-kernel-threshold-ds}, the first step of \textit{SolveTDS} transforms the instance~$(G,k)$ of \textsc{Threshold Dominating Set} into an equivalent instance~$(G',k)$ of \BW{} which will be solved by \textit{Branch}. 

The input of \textit{Branch} is a bw-graph $G$, $k' \in \mathds{N}$, and a partial solution $D \subseteq V(G)$.
The goal is to find a \bw{} $D'$ such that $D \subseteq D'$ and $|D'| \le |D| + k'$.
First, in Line~\ref{line-update-demands}, we count how many times each black vertex~$v$ is dominated by~$D$. 
If~$v$ is dominated by~$D$ at least~$r$ times, we recolor~$v$ white. Subsequently, every black vertex needs to be dominated at least once. Thus, we return No when there is at least one black vertex but the remaining budget~$k'$ is~$0$. 
Then, we continuously compute a greedy~$2$-maximal independent set~$I$ in~$G[B]$, the subgraph induced by the black vertices. 
Recall that a maximal independent set $I$ is $2$-maximal in $G$, if it holds for every $x, y\in V(G) \setminus I$ and $v \in I$ that $(I \setminus \{ v \}) \cup \{ x, y \}$ is not an independent set.

If~$|I|> k'$, \textit{Branch} computes a subset~$I'\subseteq I$ of size exactly~$k'+1$. 
Let $P$ be the set of vertices with at least two neighbors in~$I'$. 
Now, for each vertex~$v\in P\setminus D$, \textit{Branch} branches in Line~\ref{line-choose-v-in-p} to add $v$ to the partial solution~$D$.

Otherwise,~$|I|\le k'$. 
Observe that  it is not sufficient to add~$I$ to~$D$---for example, it is possible that $|N[v]\cap D|<r-1$ for some $v\in I$. 
Indeed, it may be necessary to include some vertices of~$W$ into~$D$. 
To deal with such cases, our algorithm resorts to brute force, trying all possibilities of choosing a subset of at most~$k'$ vertices of~$V(G)\setminus D$.
In order to achieve the claimed running time, however, we have to ensure that the number of vertices is bounded.
Recall that \Cref{lemma-bound-number-white-vertices} gives an upper bound on the number of white vertices in terms of the number of black vertices, provided that \Cref{rr-remove-white-vertices} has been applied exhaustively.
Therefore, it suffices to bound the number of black vertices.
To achieve this, we rely on \Cref{rr:clique}:
\textit{Branch} checks if the private black neighborhood~$B_v:=B\cap(N(v)\setminus N(I\setminus\{v\}))$ of a vertex~$v\in I$ has size at least~$ck$. 
If so, it applies \Cref{rr:clique} on~$B_v\cup\{v\}$ (Line~\ref{line-apply-rr-clique}) and iterates the procedure on a new greedy~$2$-maximal independent set~$I$ in the induced subgraph~$G[B]$ of the black vertices. 
Otherwise,~$|B_v| < ck$ for each vertex in~$I$.
As we will show later, this ensures that the number of black vertices is upper-bounded.







In the following, we will show that Algorithm~\ref{algo-branching-threshold-ds} is correct and has a running time of~$\Oh^*\big((ck)^{\Oh(rk)}\big)$.

\begin{lemma}
\label{lemma-branch-correct}
Algorithm \textit{SolveTDS} is correct. 
\end{lemma}

\begin{proof}
Line~\ref{line-transform-into-wb-tds} is correct since it is the above-mentioned reduction of \textsc{Threshold Dominating Set} to \textsc{BW-Threshold Dominating Set}.
Furthermore, Line~\ref{line-initial-branching} is correct since each solution is a superset of the empty set.
 Hence, it suffices to show that \textit{Branch}$(G,k',D)$ is correct. 
If a vertex is dominated at least~$r$ times by~$D$ we may safely color it white (Line~\ref{line-update-demands}) since \textit{Branch} returns only supersets of~$D$.  
Moreover, an instance with at least one black vertex is clearly a No-instance already if $k' = 0$ (Line~\ref{line-b-zero-k-one}). 
In the following, let~$I$ be the~$2$-maximal independent set of~$G[B]$ computed in Line~\ref{line-compute-greedy-is}. 

Suppose that~$|I|\ge k'+1$. 
Let~$I'\subseteq I$ with~$|I'|=k'+1$ be the set chosen in Line~\ref{line-compute-set-size-k-plus-one} and let~$P$ be the set of vertices in~$G$ which have at least two neighbors in~$I'$ computed in Line~\ref{line-compute-common-neigs-of-i}. 
Since~$|{}I'|=k'+1$, the sought solution must contain at least one vertex vertex of~$P\setminus D$.
Hence, the branching in Line~\ref{line-choose-v-in-p} is correct. 

Suppose that~$|I|\le k'$.  
First, assume that the private black neighborhood~$B_x:=B\cap(N(x)\setminus N(I\setminus\{x\}))$ has size at least~$ck$ (Line~\ref{line-many-priviate-black-neighbors}).
Note that~$B_x$ is a clique by \Cref{lemma:privateclique}.
Since~$|B_x|>ck$, we can apply \Cref{rr:cliquei} in Line~\ref{line-apply-rr-clique}.
Second, assume that~$B_x$ is smaller than~$ck$ for every vertex~$x\in I$. 
By \Cref{lemma-rr-remove-white-vertices-correct}, \Cref{rr-remove-white-vertices} is correct and hence Line~\ref{line-apply-rr-remove-white} is correct. 
Since the algorithm considers each possibility of adding at most~$k'$ vertices to the partial solution~$D$, \textit{Branch} is clearly correct in this case. 

Hence, \textit{SolveTDS} finds a threshold dominating set~$D$ of size at most~$k$ for the input instance~$(G,k)$ of \textsc{Threshold Dominating Set} if one exists.
\end{proof}

Before we analyze the running time of \textit{SolveTDS}, we bound the number of vertices of those graphs that are solved by the brute-force search in the algorithm. 

\begin{lemma}
\label{lemma-bound-sizes-black-and-white-after-rr-white}
After the exhaustive application of \Cref{rr-remove-white-vertices} in Line~\ref{line-apply-rr-remove-white}, the graph~$G$ contains~$\Oh(ck^2)$ black vertices and~$\Oh(c^3k^4+c^{r-1}k^{2r-2})$ white vertices. 
\end{lemma}
\begin{proof}
First, we bound the number of black vertices. 
Let~$I$ be the~$2$-maximal independent set of~$G[B]$ obtained in Line~\ref{line-compute-greedy-is}. In Line~\ref{line-apply-rr-remove-white}, we may assume that~$|I|\le k' \le k$. 
Observe that, since~$I$ is maximal, each black vertex has at least one neighbor in~$I$. 
Since $I$ is an independent set, there are at most $(c - 1) \cdot \binom{k}{2}$ vertices which have at least two neighbors in~$I$.

It remains to bound the number of black vertices with exactly one neighbor in~$I$. 
Let~$B_x := B \cap (N(x) \setminus N(I \setminus \{ x \}))$ be the private black neighborhood of some vertex~$x \in I$. 
By \Cref{lemma:privateclique}, $B_x$ is a clique.
Since \textit{Branch} is currently in Line~\ref{line-apply-rr-remove-white} we conclude that~$|B_x|<ck$ for each~$x\in I$ due to the check in Line~\ref{line-many-priviate-black-neighbors}.
Thus, overall there are at most~$ck^2$ private black vertices.
Hence, the total number of black vertices is bounded by~$\Oh(ck^2)$. 

For white vertices, note that \Cref{rr-remove-white-vertices} has been applied exhaustively.
Thus, we conclude from \Cref{lemma-bound-number-white-vertices} that the overall number of white vertices is $\Oh(c |B|^2+|B|^{r-1}) = \Oh(c^3 k^4 + c^{r - 1} k^{2r - 2})$. 
Hence, the claimed bound follows. 
\end{proof}

Now, we examine the running time of \textit{SolveTDS}.

\begin{lemma}
\label{lemma-running-time-fpt-threshold-ds}
Algorithm \textit{SolveTDS} runs in~$\Oh^*\big((ck)^{\Oh(rk)}\big)$ time.
\end{lemma}

\begin{proof}
The equivalent instance of \BW{} can be constructed in linear time. 
Next, we analyze the time complexity of \textit{Branch}.
It is easy to see that a 2-maximal independent set $I$ can be computed in polynomial time.

First, we show that \textit{Branch} either enters Line~\ref{line-inner-branching} or Line~\ref{line-brute-force} in polynomial time.
In other words, we show that \textit{Branch} spends polynomial time before it calls \textit{Branch} recursively or returns Yes or No determined by a brute-force search. 
Note that since \Cref{rr:clique,rr-remove-white-vertices} can be applied in polynomial time, Lines~\ref{line-compute-greedy-is}--\ref{line-apply-rr-remove-white} can be applied in polynomial time as well.
Hence, if the~$2$-maximal independent set~$I$ computed in Line~\ref{line-compute-greedy-is} has size at least~$k'+1$, \textit{Branch} enters Line~\ref{line-inner-branching} in polynomial time.
Thus, in the following we assume that~$I$ has size at most~$k'$ and hence Line~\ref{line-else} is entered.
Now, observe that each application of \Cref{rr:clique} reduces the number of black vertices by at least~$ck-1$. 
Since~$G$ contains at most~$n$ black vertices, Line~\ref{line-many-priviate-black-neighbors} is thus executed at most~$n$ times.
Hence, Line~\ref{line-brute-force} is entered in polynomial time unless \textit{Branch} enters Line~\ref{line-inner-branching} first.
To conclude, \textit{Branch} reaches either Line~\ref{line-inner-branching} or Line~\ref{line-brute-force} in polynomial time.

Second, we distinguish the cases that \textit{Branch} enters Lines~\ref{line-inner-branching} or~\ref{line-brute-force}.
If \textit{Branch} enters Line~\ref{line-inner-branching} then $|I| \ge k' + 1$. 
We branch on each vertex $v \in P$ (Line~\ref{line-choose-v-in-p}).
Note that $|P| \le (c - 1) \cdot \binom{k + 1}{2} \in \Oh(ck^2)$, because each vertex in $P$ has at least two neighbors in $I'$ and~$G$ is~$c$-closed.
Otherwise, \textit{Branch} enters Line~\ref{line-brute-force}.
Here, \textit{Branch} finds a \bw{} of size at most $k$ or returns No without further branching.
Since there are $\Oh(c^3k^4+c^{r-1}k^{2r-2})$ vertices by \Cref{lemma-bound-sizes-black-and-white-after-rr-white}, trying all possibilities of choosing at most $k$ vertices requires $\Oh\big((ck)^{\Oh(rk)}\big)$ time.

Consider the search tree where each node corresponds to an invocation of \textit{Branch}. 
By the above argument, each node has~$\Oh(ck^2)$ children. 
Moreover, the depth of the search tree is at most $k$. 
Since the running time of an inner node is polynomial and the running time of a leaf node is~$\Oh^*\big((ck)^{\Oh(rk)}\big)$, the overall running time of \textit{SolveTDS} is~$\Oh^*\big((ck^2)^k\cdot (ck)^{\Oh{(rk)}}\big) = \Oh^*\big((ck)^{\Oh(rk)}\big)$.
\end{proof}

From \Cref{lemma-branch-correct,lemma-running-time-fpt-threshold-ds}  we obtain the following. 

\begin{theorem}
  \textsc{Threshold Dominating Set} can be solved in~$\Oh^*\big((ck)^{\Oh(rk)}\big)$ time.
\end{theorem}


Since \textsc{Dominating Set} is the special case~$r=1$ we obtain the following.

\begin{corollary} \label{cor:ds}
\textsc{Dominating Set} can be solved in~$\Oh^*\big((ck)^{\Oh(k)}\big)$~time. 
\end{corollary}

Note that the running time in \Cref{cor:ds} amounts to $\Oh^*(k^{\Oh(k)})$ when $c = \Theta(k)$.
In the following, we show that there is only small room for improvement in the running time of \Cref{cor:ds} for $c = \Theta(k)$.
In particular, our result shows that an algorithm that runs in $\Oh^*(k^{o(k)})$ time for $c = \Theta(k)$ would violate the Exponential Time Hypothesis (ETH) of~\cite{IPZ01}.
Recall that the ETH asserts that there is no $2^{o(n + m)}$-time algorithm for \textsc{3-SAT}, where~$n$ and~$m$ denote the number of variables and clauses, respectively.

\begin{theorem}
  \textsc{Dominating Set} cannot be solved in $\Oh^*(\ell^{o(\ell)})$ time for $\ell = c + k$, unless the ETH fails.
\end{theorem}
\begin{proof}
  We reduce from the following auxiliary problem introduced by~\cite{LMS18}:
  \problemdef
  {$k \times k$ Hitting Set}
  {A set family $\mathcal{F}$ over universe $U = [k] \times [k]$ for $k \in \mathds{N}$.}
  {Is there a subset $X \subseteq [k] \times [k]$ of size exactly $k$ such that $X \cap S \ne \emptyset$ for every $S \in \mathcal{F}$ and $|X \cap (\{ i \} \times [k])| = 1$ for every $i \in [k]$?}
  Informally speaking, the problem asks for a hitting set that contains exactly one element from every row, that is, from~$\{ i \} \times [k]$ for each~$i\in [k]$.
  \cite{LMS18} showed that unless the ETH fails, \textsc{$k \times k$ Hitting Set} cannot be solved in $\Oh^*(k^{o(k)})$ time even if every set in $\mathcal{F}$ contains at most $k$ elements.

  For an instance $(\mathcal{F}, k)$ of \textsc{$k \times k$ Hitting Set} where $|S| \le k$ for every $S \in \mathcal{F}$, we can construct in polynomial time an equivalent instance~$(G,k)$ of \textsc{Dominating Set} on $(k + 1)$-closed graphs as follows:
  The vertex set $V(G)$ is $U \cup \mathcal{F} \cup R$, where~$R = \{ r_1, \dots, r_k \}$.
  We add edges such that $U$ forms a clique in $G$.
  Furthermore, we add an edge between $u \in U$ and $S \in \mathcal{F}$ if and only if $u \in S$.
  Finally, we add an edge between~$(i, j) \in U$ and $r_i$ for every~$i, j \in [k]$.
  To see why $G$ is~$(k + 1)$-closed, note that every vertex in $\mathcal{F}$ and $R$ has at most $k$ neighbors and that~$G[U]$ is a clique.
  
  We show that the instance $(\mathcal{F}, k)$ has a solution $X$ if and only if $G$ has a dominating set~$D$ of size $k$.

  For the forward direction, observe that every solution $X$ of $(\mathcal{F}, k)$ dominates every vertex in $G$ by construction.
  Conversely, suppose that the constructed \textsc{Dominating Set} instance~$(G, k)$ is a Yes-instance.
  Since every vertex in $\mathcal{F} \cup R$ is simplicial, we can assume by \Cref{lemma-simplical-vertex} that $G$ has a dominating set $D \subseteq U$.
  By construction, $D$~has a nonempty intersection with every set in $\mathcal{F}$.
  Moreover, since $D$ dominates $R$, it follows that $D$ contains exactly one vertex of each row~$\{ i \} \times [k]$.

  Suppose that there is an algorithm that solves dominating set in $\Oh^*(\ell^{o(\ell)})$.
  Then, \textsc{$k \times k$ Hitting Set} can be solved in $\Oh^*(k^{o(k)})$ time:
  First, we construct an equivalent \textsc{Dominating Set} instance with $\ell = c + k = 2k + 1$ as described above.
  Second, we then make use of the $\Oh^*(\ell^{o(\ell)})$-time algorithm for \textsc{Dominating Set}.
  This would refute ETH.
\end{proof}



\subsection{Smaller Kernel in Bipartite Graphs}

In this subsection, we show that \textsc{Dominating Set} admits a kernel with~$\Oh(c^3 k^4)$~vertices when the input graph is bipartite and~$c$-closed.
As in \Cref{sec-kernel-threshold-ds}, we use \BW{} to show this kernelization. 
Since we only focus on the case~$r=1$, we will refer to \BW{} as \textsc{BW-Dominating Set}. 

\begin{rrule}
  \label{rr:ds:largedeg}
  Suppose that there is a vertex $v \in V(G)$ with at least $(c-1)k+1$ black neighbors.
  Then,
  \begin{enumerate}
    \item
      color all the vertices of $N_G(v)$ white,
    \item
      remove $v$, and
    \item
      decrease $k$ by 1.
  \end{enumerate}
\end{rrule}

\begin{lemma}
  \Cref{rr:ds:largedeg} is correct.
\end{lemma}
\begin{proof}
  Let $v$ be a vertex with~$|N_G(v)| \ge (c-1)k+1$ and let~$(G', k' = k - 1)$ be the instance obtained as a result of \Cref{rr:ds:largedeg}.
  Suppose that~$(G, k)$ is a Yes-instance with a bw-dominating set~$D$ of~$k$ vertices.
  We claim that~$v \in D$.
  Let~$P, Q$ be a bipartition of~$G$ and assume without loss of generality that~$v \in P$.
  Since~$G$ is~$c$-closed, each vertex in~$P \setminus \{ v \}$ can dominate at most~$c - 1$ neighbors of~$v$.
  Moreover, each vertex in~$Q$ can dominate no neighbor of~$v$ except for itself.
  It follows that~$v \in D$, since otherwise some neighbor of~$v$ is nonadjacent to all vertices of~$D$.
  Consequently,~$D \setminus \{ v \}$ is a solution of~$(G', k')$.

  Conversely, suppose that~$D'$ is a bw-dominating set of size at most~$k'$ in~$G'$.
  Then,~$(G, k)$ is a Yes-instance with a solution~$D \cup \{ v \}$. 
  
  Note that the~$c$-closure is maintained by \Cref{obs:removev}. 
\end{proof}

After \Cref{rr:ds:largedeg} is exhaustively applied, each vertex dominates at most $ck - 1$ black vertices.
Thus, we obtain the following reduction rule.

\begin{rrule}
  \label{rr:ds:manyblack}
  If there are at least $ck^2$ black vertices, return No.
\end{rrule}

Finally, we deal with white leaf vertices.
The correctness trivially follows.

\begin{rrule}
  \label{rr:ds:leaves}
  If there is a white vertex $v$ that is adjacent to only one black vertex, then remove $v$.
\end{rrule}

\begin{theorem}
  \textsc{Dominating Set} in bipartite graphs has a kernel with $\Oh(c^3 k^4)$ vertices.
\end{theorem}
\begin{proof}
  Let $(G, k)$ be an instance of \textsc{BW-Dominating Set} in which \Cref{rr:ds:largedeg,rr:ds:manyblack,rr:ds:leaves} have been exhaustively applied.
  Note that this can be done in polynomial time.

  The graph~$G$ contains at most $ck^2$ black vertices by \Cref{rr:ds:manyblack}.
  Moreover, each white vertex of~$G$ is adjacent to at least two black vertices by \Cref{rr:ds:leaves}.
  Since the graph is bipartite and $c$-closed, any pair of vertices have at most $c-1$ common neighbors.
  Consequently,~$G$ contains at most $ck^2 + c \binom{ck^2}{2} \in \Oh(c^3 k^4)$ vertices.
\end{proof}

\section{Induced Matching}
\label{sec:im}
In this section, we develop kernelizations for \textsc{Induced Matching} in $c$-closed graphs.

\problemdef
{Induced Matching}
{A graph $G$ and $k \in \mathds{N}$.}
{Is there a set $M$ of at least $k$ edges such that endpoints of distinct edges in~$M$ are pairwise nonadjacent?}

 \textsc{Induced Matching} is W[1]-hard when parameterized by $k$, even in bipartite graphs~\cite{MS09}.
In terms of kernelizations, \textsc{Induced Matching} admits a kernel with~$\Oh(\Delta^2 k)$ vertices \cite{MS09} and $\Oh(k^{d})$ vertices \cite{EKKW10,KPSX11}.
Recall that~$\Delta$ is the maximum degree and that~$d$ is the degeneracy.
The latter kernelization result is essentially tight: 
Unless coNP $\subseteq$ NP/poly, \textsc{Induced Matching} has no kernel of size $\Oh(k^{d - 3 - \varepsilon})$ for any $\varepsilon > 0$~\cite{CGH17}.
Despite the lower bound in degenerate graphs, we discover in this section that \textsc{Induced Matching} in $c$-closed graphs has a polynomial kernel when parameterized by $k + c$.

\subsection{Ramsey-like Bounds for Induced Matchings}

Dabrowski et al.~\cite{DDL13} derived fixed-parameter tractability for \textsc{Induced Matching} in $(K_q, K_{b, b})$-free graphs.
At the heart of their algorithm lies a Ramsey-type result for induced matchings:
For~$q, b \in \mathds{N}$, there exists an integer $Q_{q, b}$ such that any bipartite graph with a matching of size at least $Q_{q, b}$ contains a biclique $K_{q, q}$ or an induced matching of size~$b$.
In this subsection, we present analogous results for $c$-closed graphs where the number~$Q_{q, b}$ is polynomial in $q$ and $b$.
We begin with two preliminary lemmas.

\begin{lemma}
  \label{lemma:ramseyd}
  Any graph $G$ with a matching $M$ of size at least $2 \Delta b$ has an induced matching of size $b$.
\end{lemma}
\begin{proof}
  We prove the statement by induction on $b$.
  The lemma clearly holds for the base case~$b = 0$.
  For~$b > 0$, let $uv$ be a matched edge in $M$ and let $G' := G - N[\{ u, v \}]$.
  Since $|N[\{ u, v \}]| \le 2 \Delta_G$, there is a matching of size at least~$2\Delta_G b - 2\Delta_G \ge 2\Delta_{G'} (b - 1)$ in $G'$.
  Consequently, there is an induced matching $M'$ of size $b - 1$ in $G'$ by induction hypothesis.
  Thus, $G$~has an induced matching~$M' \cup \{ uv \}$ of size~$b$.
\end{proof}
\begin{lemma}
  \label{lemma:ramseyc}
  Suppose that $G$ is a $c$-closed bipartite graph.
  If there are at least $2b$ vertices of degree at least $cb$, then $G$ contains an induced matching of size at least $b$.
\end{lemma}
\begin{proof}
  Let $A, B$ be a bipartition of $G$.
  Without loss of generality, assume that $A$ contains a set $A'$ of exactly $b$ vertices of degree at least $cb$.
  Since $G$ is $c$-closed, $|N(v) \cap N(v')| < c$ for all $v, v' \in A'$.
  Consequently, each $v \in A'$ has a neighbor~$u$ such that~$u \notin N(v')$ for all~$v' \in A' \setminus \{ v \}$.
  Thus,~$G$ contains an induced matching of size~$b$.
\end{proof}

The following lemma shows that $c$-closed bipartite graphs with a large matching also have a large induced matching.

\begin{lemma}
  \label{lemma:ramseyb}
  Let~$\ramseyb(b) := 2cb^2 + 2b \in \Oh(cb^2)$.
  Let~$G$ be a~$c$-closed bipartite graph.
  If~$G$ has a matching~$M$ of size at least~$\ramseyb(b)$, then~$G$ contains an induced matching of size at~least~$b$.
\end{lemma}
\begin{proof}
  If there are at least~$2b$ vertices of degree at least~$cb$ in~$G$, then \Cref{lemma:ramseyc} yields an induced matching of size~$b$.
  Thus, we can assume that $|S| < 2b$ for the set $S$ of vertices of degree at least $cb$.
  Observe that $G - S$ has a matching of size $2cb^2$ and that $\Delta_{G - S} \le cb$.
  Thus, $G - S$ has an induced matching of size $b$ by \Cref{lemma:ramseyd}.
\end{proof}

We extend \Cref{lemma:ramseyb} to non-bipartite $c$-closed graphs in the subsequent two lemmas.
Recall that each~$c$-closed graph~$G$ with at least~$\ramsey(q, b) \in \Oh(cb^2 + qb)$ vertices contains a clique of $q$ vertices or an independent set of~$b$ vertices by \Cref{lemma:ramsey}.
We may now use the Ramsey bound of \Cref{lemma:ramsey} as follows. When~$G$ contains a large matching~$M$ and no large clique, then the set of endpoints of~$M$ must contain two large independent sets which induce a bipartite graph with a large matching. To this bipartite graph, we may then apply \Cref{lemma:ramseyb} and conclude that~$G$ contains a large induced matching. 

\begin{lemma}
  \label{lemma:ramseyb1}
  Let $\ramseyb'(q, b) := \ramsey(q, \ramseyb(b)) \in \Oh(c q b^2 + c^3 b^4)$.
  Any $c$-closed graph~$G$ with an independent set $I$ of size at least $\ramseyb'(q, b)$ and a matching $M$ saturating $I$ contains a clique of size $q$ or an induced matching of size $b$.
\end{lemma}
\begin{proof}
  Suppose that $G$ contains no clique of size $q$.
  We show that there is an induced matching of size $b$ in $G$.
  Let~$S := V(M) \setminus I$ be the set of vertices matched to~$I$ in~$M$.
  Since~$|S| \ge \ramsey(q, \ramseyb(b))$, it follows from \Cref{lemma:ramsey} that there is an independent set $S' \subseteq S$ of size at least $\ramseyb(b)$ in $G'$.
  Let $I' \subseteq I$ be the set of vertices matched to $S'$ in $M$ and observe that~$G[S' \cup I']$ is bipartite.
  Then, there is an induced matching of size at least~$b$ in $G[S' \cup I']$ by \Cref{lemma:ramseyb}.
  Thus, $G$ contains an induced matching of size~$b$.
\end{proof}
\begin{lemma}
  \label{lemma:ramseyb2}
  Let $\ramseyb''(q, b) := \ramsey(q, \ramseyb'(b)) \in \Oh(c^3 q^2 b^4 + c^7 b^8)$.
  Any $c$-closed graph $G$ with a matching $M$ of size at least~$\ramseyb''(q, b)$ contains a clique of size $q$ or an induced matching of size $b$.
\end{lemma}

\begin{proof}
  Suppose that $G$ contains no clique of size $q$.
  We will show that there is an induced matching of size $b$ in $G$.
  Let~$I$ and~$S$ be disjoint vertex sets such that~$I$ and~$S$ consist of distinct endpoints of each edge in~$M$.
  Since~$|I| \ge \ramsey(q, \ramseyb'(b))$, it follows from \Cref{lemma:ramsey} that there is an independent set~$I' \subseteq I$ of size~$\ramseyb'(b)$.
  Let $S' \subseteq S$ be the set of vertices matched to $I'$ and let $G' := G[S' \cup I']$.
  Since $I$ is an independent set of size at least $\ramseyb'(b)$, it follows from \Cref{lemma:ramseyb} that there is an induced matching~$M'$ of size at least $b$ in $G'$.
  Consequently, $G$ contains an induced matching of size~$b$.
\end{proof}

\subsection{A Polynomial Kernel in \texorpdfstring{$c$}{c}-closed Graphs}

In this subsection, we prove that \textsc{Induced Matching} in $c$-closed graphs admits a kernel with $\Oh(c^{7} k^{8})$ vertices.
Our kernelization is based on \Cref{lemma:ramseyb1,lemma:ramseyb2}.
To utilize these lemmas, we start with a reduction rule that destroys large cliques.

\begin{rrule}
  \label{rr:im:vlargevc}
  Let $v \in V(G)$ and let $M_v$ be a maximum matching in $G[N(v)]$.
  If~$|M_v| \ge 2ck$, then remove $v$.
\end{rrule}
\begin{lemma}
  \Cref{rr:im:vlargevc} is correct.
\end{lemma}
\begin{proof}
  Let $v \in V(G)$, let $M_v$ be a maximum matching in $G[N_G(v)]$ of size at least $2ck$, and let $G' := G - v$.
  Suppose that $G$ has an induced matching $M$ of size at least $k$.
  We show that $G'$ contains an induced matching of size at least $k$ as well.
  We are done if $M$ does not use $v$, because then $M$ is also an induced matching in~$G'$.
  Hence, assume that $M$ uses $v$ in the following.
  Let $v_1 v_2, \dots, v_{2k - 1} v_{2k}$ be $k$ edges of $M$ such that $v_{2k} = v$.
  By the definition of induced matching, $v_i \notin N_G(v)$ for each $i \in [2k - 2]$.
  Thus, the $c$-closure of $G$ yields that $|N_G(v) \cap N_G(v_i)| < c$ for each~$i \in [2k - 2]$.
  Since $M_v$ has size at least $2ck$, there is an edge~$e$ in $M_v$ whose endpoints are not adjacent to any vertex $v_i$ for $i \in [2k - 2]$.
  Hence, the edges~$v_1 v_2, \dots,v_{2k - 3} v_{2k - 2}, e$ form an induced matching of size $k$ in $G'$.
  The other direction follows trivially.
  Note that the $c$-closure is maintained by \Cref{obs:removev}.
\end{proof}

Henceforth, we assume that \Cref{rr:im:vlargevc} has been applied for each vertex.
In the next lemma, we verify that there is no large clique.

\begin{lemma}
  \label{lemma:nolargeclique}
  There is no clique of size $4ck + 1$ in $G$.
\end{lemma}

\begin{proof}
  Suppose that $G$ contains a clique $C$ of size at least $4ck + 1$ and let $v \in C$.
  Then, there is a matching of size $2ck$ in $C \setminus \{ v \} \subseteq N(v)$, contradicting the fact that \Cref{rr:im:vlargevc} has been applied on every vertex.
\end{proof}

Once we show that the graph has a sufficiently large matching, \Cref{lemma:ramseyb2} tells us that we can find a sufficiently large induced matching as well.
Note, however, that a graph may not have a sufficiently large matching, even if it contains many vertices (consider a star $K_{1, n}$ with $n$ leaves).
Our way around this obstruction is the LP (Linear Programming) relaxation of \textsc{Vertex Cover} (henceforth, we will abbreviate it as VCLP).
It is well-known in the theory of kernelization that VCLP almost trivially yields a linear-vertex kernel for \textsc{Vertex Cover}~\cite{CKJ01} due to the Nemhauser-Trotter theorem \cite{NT75}.
Here, we will exploit VCLP to ensure that after we apply some reduction rules, either the size of $G$ is upper-bounded or the minimum vertex cover size (or equivalently the maximum matching size) of $G$ is sufficiently large.

Recall that \textsc{Vertex Cover} can be formulated as an integer linear program as follows, using a variable~$x_v$ for each~$v \in V(G)$:
\begin{align*}
    \min \sum_{v \in V(G)} x_v \qquad\text{subject to}\quad & x_u + x_v \ge 1 \quad \forall uv \in E(G),  \\[-3ex]
                                                       & x_v \in \{ 0, 1 \} \quad \forall v \in V(G).
\end{align*}

In VCLP, the last integral constraint is relaxed to $0 \le x_v \le 1$ for each $v \in V(G)$.
It is known that VCLP admits a half-integral optimal solution (that is, $x_v \in \{ 0, 1/2, 1 \}$ for each~$v \in V(G)$) and such a solution can be computed in $\Oh(m \sqrt{n})$~time via a reduction to \textsc{Maximum Matching} (see, for instance,~\cite{BE83}~or~\cite[Section 2.5]{CFK+15}).
Suppose that we have a half-integral optimal solution $(x_v)_{v \in V(G)}$.
Let~$V_0 := \{ v \in V(G) \mid x_v = 0 \}$,~$V_1 := \{ v \in V(G) \mid x_v = 1 \}$, and~$V_{1/2} := \{ v \in V(G) \mid x_v = 1 / 2 \}$.

We will bound the sizes of $V_0$, $V_1$, and $V_{1/2}$ in the upcoming rules.
We begin with~$V_{1/2}$.
We use the bound $\ramseyb''$ as specified in \Cref{lemma:ramseyb2}.

\begin{rrule}
  \label{rr:vhalf}
  If $|V_{1/2}| \ge 3 \ramseyb''(4ck + 1, k)$, then return Yes.
\end{rrule}

To show the correctness, we will use the fact that $VC + MM \ge 2 LP$ for any graph $G$~\cite[Lemma 2.1]{GP16}.
Here, $VC$, $MM$, and $LP$ refer to the minimum vertex cover size, the maximum matching size, and the optimal VCLP cost of $G$.

\begin{lemma}
  \Cref{rr:vhalf} is correct.
\end{lemma}
\begin{proof}
  Observe that the optimal cost of VCLP for $G[V_{1/2}]$ is $|V_{1/2}| / 2$.
  Let $X$ be a minimum vertex cover and $M$ be a maximum matching in $G[V_{1/2}]$.
  Then, it follows that $|X| + |M| \ge |V_{1/2}|$~\cite[Lemma 2.1]{GP16}.
  Since $V(M)$ is a vertex cover in $G[V_{1/2}]$, we also have $2|M| \ge |X|$.
  Thus, $|M| \ge |V_{1/2}| / 3 \ge \ramseyb''(4ck+1, k)$.
  Recall that there is no clique of size $4ck + 1$ by \Cref{lemma:nolargeclique}.
  Hence, \Cref{lemma:ramseyb2} yields that $G$ contains an induced matching of size at least~$k$.
\end{proof}

We next upper-bound the size of $V_{1}$.
See \Cref{lemma:ramseyb1} for the definition of~$\ramseyb'$.
\begin{rrule}
  \label{rr:vone}
  If $|V_1| \ge \ramseyb'(4ck + 1, k)$, then return Yes.
\end{rrule}

To prove the correctness of \Cref{rr:vone}, let us introduce the notion of \emph{crowns}~\cite{CFJ04}.
For a graph $G$, a crown is an ordered pair $(I, H)$ of vertex sets of $G$ with the following properties:

\begin{enumerate}
  \item 
    $I \ne \emptyset$ is an independent set in $G$,
  \item
    $H = N(I)$, and
  \item
    there is a matching saturating $H$ in $G[H, I]$.
\end{enumerate}

Crowns are closely related to VCLP---in fact, $(V_0, V_1)$ is a crown~\cite{AFLS07,CC08}.

\begin{lemma}
  \Cref{rr:vone} is correct.
\end{lemma}
\begin{proof}
  Since $(V_0, V_1)$ is a crown in $G$, there is a matching $M$ saturating $V_1$ in~$G[V_0, V_1]$.
  By definition, $I := V_0 \cap V(M)$ is an independent set of size $|V_1| \ge \ramseyb'(4ck + 1, k)$ in $G$.
  Now, it follows from \Cref{lemma:ramseyb1} that $G[I \cup V_1]$ contains an induced matching of size $k$.
\end{proof}

To deal with $V_0$, we introduce some additional rules which may add or remove vertices.
Let us start with a simple rule.
Basically, if there are multiple leaf vertices with the same neighborhood, then only one of them is relevant.

\begin{rrule}
  \label{rr:removeleaftwin}
  If $v_1 \in V_1$ has more than one leaf neighbor, then remove all but one of them.
\end{rrule}

The correctness of \Cref{rr:removeleaftwin} is obvious and thus we omit the proof.

\begin{rrule}
  \label{rr:crownleaf}
  Let $v_0 \in V_0$ and let $v_1 \in V_1$.
  If $N_G[v_0] \subseteq N_G[v_1]$ and there is no leaf vertex attached to $v_1$, then attach a leaf vertex $\ell$ to $v_1$.
\end{rrule}

\begin{lemma}
  \Cref{rr:crownleaf} is correct.
\end{lemma}
\begin{proof}
  Let $G'$ be the graph obtained by adding a leaf vertex $\ell$ to $v_1$. We show that~$(G,k)$ is a yes-instance if and only if~$(G',k)$ is.
  The forward direction is trivial.
  For the other direction, note that any induced matching $M'$ in $G'$ is an induced matching in $G$ if $M'$ does not include $v_1 \ell$.
  Hence, it suffices to show that if there is an induced matching $M'$ in $G'$ such that $|M'| \ge k$ and $v_1 \ell \in M'$, then there is an induced matching of size $k$ in $G$ as well.
  By the definition of induced matching, $M' \setminus \{ v_1 \ell \}$ contains no edge that is incident with a neighbor of $v_1$.
  Since $N_G[v_0] \subseteq N_G[v_1]$, the same holds for~$v_0$.
  Thus,~$(M' \setminus \{ v_1 \ell \}) \cup \{ v_0 v_1 \}$ is an induced matching of size at least~$k$ in $G$.
\end{proof}

For $c > 1$, \Cref{rr:crownleaf} maintains the $c$-closure  by \Cref{obs:addv}.
Note that \textsc{Induced Matching} can be solved in linear time when $G$ is 1-closed:
In this case, $G$ is a disjoint union of complete graphs, and $(G, k)$ is a Yes-instance if and only if $G$ contains at least $k$ cliques of size at least two.

\begin{rrule}
  \label{rr:removetwin}
  Let $v_0 \in V_0$ be a non-leaf vertex.
  If each vertex $v_1 \in N_G(v_0)$ has a leaf neighbor, then remove $v_0$.
\end{rrule}

\begin{lemma}
  \Cref{rr:removetwin} is correct.
\end{lemma}
\begin{proof}
  Let $G' = G - v_0$.
  Suppose that $G$ has an induced matching $M$ of size at least $k$.
  If~$M$ does not use $v_0$ we are done.
  So assume that $M$ includes $v_0 v_1$ for~$v_1 \in N_G(v_0)$.
  Since there is a leaf vertex $\ell$ attached to $v_1$, the set $(M \setminus \{ v_0 v_1 \}) \cup \{ v_1 \ell \}$ is an induced matching of size at least $k$ in $G'$.
  The other direction follows trivially.
  The $c$-closure is maintained by \Cref{obs:removev}.
\end{proof}

\begin{theorem}
  \label{thm:impk}
  \textsc{Induced Matching} has a kernel with $\Oh(c^{7} k^{8})$ vertices.
\end{theorem}

\begin{proof}
  We apply \Cref{rr:im:vlargevc,rr:vhalf,rr:vone,rr:removeleaftwin,rr:crownleaf,rr:removetwin} exhaustively.
  We also remove all isolated vertices.
  It is easy to verify that all these rules can be exhaustively applied in polynomial time.

  Recall that every vertex $v \in V(G)$ has $x_v = 0$, $x_v = 1$, or $x_v = 1/2$ in the half-integral solution $x$ and that $V_i$ is the set of vertices $v$ with $x_v = i$ for $i \in \{ 0, 1/2, 1\}$.
  Note that~$|V_{1/2}| \in \Oh(c^{7} k^{8})$ and~$|V_1| \in \Oh(c^3 k^4)$ by \Cref{rr:vhalf,rr:vone}.
  We show that~$|V_0| \in \Oh(c |V_1|^2) = \Oh(c^7 k^{8})$.
  Note that there are at most $|V_1|$ leaf vertices in $V_0$ by \Cref{rr:removeleaftwin}.
  All other vertices in $V_0$ are adjacent to at least two nonadjacent vertices in $V_1$:
  If there exists a vertex $v_0 \in V_0$ such that $N_G(v_0)$ is a clique of size at least two, then \Cref{rr:crownleaf} adds a leaf vertex to each vertex in $N_G(v_0)$ and \Cref{rr:removetwin} removes $v_0$.
  Since $G$ is $c$-closed, there are thus~$c \binom{|V_1|}{2}$ non-leaf vertices in $V_0$.
  It follows that $|V_0| < |V_1| + c \binom{|V_1|}{2} \in \Oh(c^7 k^{8})$.
\end{proof}

\subsection{A Smaller Kernel in Bipartite \texorpdfstring{$c$}{c}-closed Graphs}


In the following, we provide smaller kernels in bipartite graphs.
Our kernelization is based on the following lemma, proven by a meet-in-the-middle approach on vertex degrees. 
Interestingly, this lemma will also play a central role in the kernelization for \textsc{Irredundant Set} in \Cref{sec:irs}.

\begin{lemma}
  \label{lemma:im:ramseyd2}
  Any bipartite graph $G$ with at least $6 \Delta^{3/2} b + 2 \Delta b$ non-isolated vertices has an induced matching of size $b$.
\end{lemma}

\begin{proof}
We show the contrapositive. Thus, assume that~$G$ has no induced matching of size~$b$. By \Cref{lemma:ramseyd}, this implies that $|M| < 2 \Delta b$ for every maximum matching of~$G$.
  By K\H{o}nig's theorem, $G$ thus has a vertex cover $X$ of size $|M|$.
  For a bipartition~$A, B$ of~$G$, let~$X_A := X \cap A$,~$Y_B := B \setminus X$, and~$G' := G[X_A, Y_B]$.
  Note that~$X_A$ is a vertex cover for~$G'$ and thus~$|E(G')|\le |X_A|\cdot\Delta$.
  We prove that if $|Y_B| > 3 \Delta^{3/2} b$, then $G'$ has an induced matching of size $b$.
  Let $Z_B := \{ y \in Y_B \mid \deg_{G'}(y) \ge \Delta^{1/2} \}$ and let $Z_B' := Y_B \setminus Z_B$.
  Since $G'$ has at most $\Delta \cdot |X_A| < 2 \Delta^2 b$ edges, we have $|Z_B| \le |X_A| / \Delta^{1/2} \le 2 \Delta^{3/2} b$.

  We show that $|Z_B'| \le \Delta^{3/2} b$ by induction on $b$.
  For the base case $b = 1$, observe that every vertex in $Y_B$ has at least one neighbor in $G'$ since $N_{G'}(y) = N_{G}(y) \ne \emptyset$ for every $y \in Y_B$.
  Thus, there is an induced matching of size one.
  For $b > 1$, let~$z$ be an arbitrary vertex from~$Z_B'$ and consider the graph $G_z'$ obtained by deleting~$N_{G'}(z)$ and $N_{G'}(N_{G'}(z))$.
  Note that we delete $\deg_{G'}(z) \le \Delta^{1/2}$ vertices from~$X_A$ and $|N_{G'}(N_{G'}(z))| \le \sum_{x \in N_{G'}(y)} \deg_{G'}(x) \le \Delta^{3/2}$ vertices from $Z_B'$.
  To use the inductive hypothesis on $G'_{z}$, we need to show that $G'_{z}$ has no induced matching of size~$b - 1$ and that every vertex in $Z_B' \setminus N_{G'}(N_{G'}(z))$ has at least one neighbor in~$G'_{z}$.
  If $G_z'$ has an induced matching $M'$ of size $b - 1$, then $M' \cup \{ xz \}$ is an induced matching for an arbitrary vertex~$x \in N_{G'}(z)$ since $N_{G'}(x) \cup N_{G'}(z)$ is absent from~$G'$, contradicting the assumption that $G$ has no induced matching of size~$b$.
  Note that for every $y \in Y_B \setminus N_{G'}(N_{G'}(z))$, we have $N_{G_z'}(y) = N_{G}(y) \ne \emptyset$ and hence $y$ has at least one neighbor in $G_z'$.
  Now, by the induction hypothesis, we have $|Z_B' \setminus  N_{G'}(N_{G'}(z))| \le \Delta^{3/2} (b - 1)$ and consequently, $|Z_B'| \le \Delta^{3/2} b$.

  Putting everything together, we have~$|X| < 2 \Delta b$ and~$|B \setminus X| = |Z_B| + |Z_B'| < 3 \Delta^{3/2} b$.
  We can analogously show that~$|A \setminus X| < 3 \Delta^{3/2} b$, thereby proving the lemma.
\end{proof}

The following theorem now follows from \Cref{lemma:im:ramseyd2}.
\begin{theorem}
  \textsc{Induced Matching} in bipartite graphs admits a kernel with $\Oh(\Delta^{3/2} k)$ vertices.
\end{theorem}

We can also use \Cref{lemma:im:ramseyd2} to obtain a smaller kernel in $c$-closed bipartite graphs.

\begin{theorem}
  \textsc{Induced Matching} in $c$-closed bipartite graphs admits a kernel with $\Oh(c^{3/2} k^{5/2})$ vertices.
\end{theorem}

\begin{proof}
  We first remove all isolated vertices.
  Moreover, as long as there is a vertex~$v$ with more than one leaf neighbor, we remove all but one (cf. \Cref{rr:removeleaftwin}).

  We show that the resulting $c$-closed bipartite graph $G$ has an induced matching of size $k$ whenever $|V(G)| \ge  6 c^{3/2} k^{5/2} + 2 c k^2 + c \binom{2k}{2} + 4k$.
  Let~$S \subseteq V(G)$ be the set of vertices whose degrees are at least~$ck$ and let~$T := V(G) \setminus S$.
  Since $|V(G)| = |S| + |T|$, either $|S| \ge 2k$ or $|T| \ge 6 c^{3/2} k^{5/2} + 2 c k^2 + c \binom{2k}{2} + 2k$ holds.
  If~$|S| \ge 2k$, then we can conclude that there is an induced matching of size~$k$ in~$G$ by \Cref{lemma:ramseyc}.
  Thus, it suffices to show that $G$ has an induced matching of size $k$ for $|S| < 2k$ and~$|T| \ge 6 c^{3/2} k^{5/2} + 2 c k^2 + c \binom{2k}{2} + 2k$.
  To do so, we will make use of \Cref{lemma:im:ramseyd2} to find an induced matching of size $k$ in $G[T]$.
  Note, however, that \Cref{lemma:im:ramseyd2} counts the number of non-isolated vertices.
  So we first show that~$G[T]$ does not contain many isolated vertices.
  More precisely, we prove that $|I| < c \binom{2k}{2} + 2k$, where $I$ is the set of isolated vertices.
  Since $G$ has no isolated vertex, we have $N_G(v) \subseteq S$ for every~$v$ in $I$.
  Every vertex in $G$ has at most one leaf neighbor and thus there are at most~$|S| \le 2k$ vertices of $I$ that have degree one in $G$.
  Furthermore, the number of vertices of $I$ that have at least two neighbors in $G$ is smaller than~$c \binom{2k}{2}$:
  For a bipartition $(S_1, S_2)$ of $S$ in $G$, we have $N(v) \subseteq S_1$ or $N(v) \subseteq S_2$ for every~$v \in I$.
  Moreover, at most~$c \binom{|S_1|}{2} + c \binom{|S_2|}{2}$ have at least two neighbors in $S_1$ or $S_2$.
  Since~$|S_1| + |S_2| = |S| < 2k$, we see that there are less than such $c \binom{2k}{2}$ vertices.
  Consequently, we have $|I| < c \binom{2k}{2} + 2k$.

  We have shown that $G[T]$ has at least $|T| - |I| \ge 6 c^{3/2} k^{5/2} + 2 c k^2$ non-isolated vertices.
  By definition~$\Delta_{G[T]} \le ck$ and thus the existence of an induced matching of size $k$ then follows from \Cref{lemma:im:ramseyd2}.
\end{proof}

\section{Irredundant Set}
\label{sec:irs}
A vertex set~$S \subseteq V(G)$ is \emph{irredundant} if there is a private neighbor for each vertex~$v$ in~$S$.
Here, a \emph{private neighbor} of~$v \in S$ is a vertex~$v' \in N[v]$ (possibly~$v' = v$) such that~$v' \notin N(u)$ for each $u \in S \setminus \{ v \}$.
Note that every independent set is also irredundant, but an irredundant set is not necessarily independent.

\problemdef
{Irredundant Set}
{A graph $G$ and $k \in \mathds{N}$.}
{Is there an irredundant set $S$ of at least $k$ vertices in $G$?}

\textsc{Irredundant Set} is W[1]-hard with respect to~$k$~\cite{DFR00} but it admits a kernel with at most~$(d + 1) k$ vertices in $d$-degenerate graphs.
This is because any $d$-degenerate graph on at least $(d + 1) k$ vertices contains an independent set and thus an irredundant set of at least $k$ vertices. Observe also that every bipartite graph on~$n$ vertices has an irredundant set of size at least~$n/2$ since every independent set is an irrendundant set. In other words, \textsc{Irredundant Set} admits a trivial~$2k$-vertex kernel on bipartite graphs. 
In this section, we show that \textsc{Irredundant Set} admits a kernel with $\Oh(c^{5/2} k^3)$ vertices.
Our kernelization relies on the Ramsey bound (\Cref{lemma:ramsey}) and the bound on induced matchings (\Cref{lemma:im:ramseyd2}).
We show that the following reduction rule suffices to obtain a polynomial kernel.

\begin{rrule}
  \label{rr:simplicialtwin}
  If $u, v \in V(G)$ are simplicial vertices such that $N_G[u] = N_G[v]$, then remove $v$.
\end{rrule}

\begin{lemma}
  \Cref{rr:simplicialtwin} is correct.
\end{lemma}
\begin{proof}
  Let $u, v \in V(G)$ be simplicial vertices such that $N_G[u] = N_G[v]$.
  Let $G'$ be the graph obtained by removing $v$ as specified in \Cref{rr:simplicialtwin}.
  Suppose that $(G, k)$ is a Yes-instance with a solution $S$.
  It must hold that $u \notin S$ or $v \notin S$ by the definition of irredundant sets.
  Without loss of generality, assume that $v \notin S$.
  If~$v$ is a private neighbor of $w \in S$ (possibly~$w = u$), then $u$ is also a private neighbor of $w$.
  Thus, $(G', k)$ is also a Yes-instance.
  The other direction follows trivially.
  The $c$-closure is maintained by \Cref{obs:removev}.
\end{proof}

We prove that \Cref{rr:simplicialtwin} yields a kernelization of the claimed size.

\begin{theorem}
  \textsc{Irredundant Set} in $c$-closed graphs has a kernel with $\Oh(c^{5/2} k^3)$ vertices.
\end{theorem}
\begin{proof}
  We assume that \Cref{rr:simplicialtwin} has been applied exhaustively.

  To simplify notation, let $\alpha' := 6 c^{3/2} k + 2ck + 1 \in \Oh(c^{3/2} k)$ and $\alpha := \ramsey(\alpha', k) \in \Oh(c^{3/2} k^2)$.
  We claim that any instance~$(G, k)$ with at least~$\ramsey(c \alpha + 1, k) \in \Oh(c^{5/2} k^3)$ vertices is a Yes-instance.
  By \Cref{lemma:ramsey},~$G$ has a clique of size~$c \alpha + 1$ or an independent set of size~$k$.
  Since any independent set is also an irredundant set, $(G, k)$ is a Yes-instance when~$G$ contains an independent set of size $k$.
  Thus, we can simply return Yes when the algorithm of \Cref{lemma:ramsey} finds an independent set.
  Suppose that our algorithm finds a clique of size $c \alpha + 1$.
  Let $C$ be the maximal clique containing it.

  It remains to show that~$(G, k)$ is a Yes-instance.
  Let~$C' = \{ v \in C \mid N_G(v) \setminus C \ne \emptyset \}$ be the set of vertices in~$C$ that have at least one neighbor outside~$C$.
  There exists at most one vertex~$v$ with $N_G[v] = C$ by \Cref{rr:simplicialtwin} and thus $|C'| \ge |C| - 1 \ge c \alpha$.
  Let $G' = G - (C \setminus C')$.
  That is,~$G'$~is a graph obtained by removing the vertex whose closed neighborhood is~$C$, if it exists.
  For each~$i \in [\alpha]$, we will choose vertices $x_i \in C'$ and~$y_i \in N_{G'}(C')$ as follows:
  Let $x_i$ be an arbitrary vertex in $C' \setminus \bigcup_{j \in [i - 1] } N_{G'}(y_j)$ and let~$y_i$ be an arbitrary vertex in $N_{G'}(x_i)$.
  Note that $C' \setminus \bigcup_{j \in [i - 1] } N_{G'}(y_j) \ne \emptyset$ for each $i \in [\alpha]$, because $|C'| \ge c \alpha$ and $y_j$ has less than~$c$~neighbors in $C'$ for all $j \in [i - 1]$ by \Cref{obs:cliqueintersection}.
  Observe that $x_i y_i \in E(G)$ for every $i \in [\alpha]$ and that $x_j y_i \notin E(G)$ for every $i < j \in [\alpha]$ by the choice of $x_i$ and $y_i$.

  We apply \Cref{lemma:ramsey} on $y_1, \dots, y_{\alpha}$ to search for an independent set of size $k$ or a clique of size $\alpha'$.
  In the former case, we again simply return Yes.
  Let $Y = \{ y_{i_1}, \dots, y_{i_{\alpha'}} \}$ be a clique of size~$\alpha'$ and let~$X = \{ x_{i_1}, \dots, x_{i_\alpha'} \}$.
  Without loss of generality, we assume that~$i_1 < i_j$ for every~$j \in [2, \alpha']$.
  For $X' = X \setminus \{ x_{i_1} \}$ and $Y' = Y \setminus \{ y_{i_1} \}$, we prove that the bipartite graph $G[X', Y']$ has an induced matching of size $k$, using \Cref{lemma:im:ramseyd2}.
  First we show that~$\Delta_{G[X', Y']} < c$.
  All vertices in $Y'$ have less than $c$ neighbors in $X'$ by \Cref{obs:cliqueintersection}.
  Since $i_1 < i_j$, we have $x_{i_j} \notin N_G(y_{i_1})$ for all $j \in [2, \alpha']$.
  It follows from the $c$-closure of~$G$ that~$x_{i_j}$ has less than $c$ neighbors in $Y' \subseteq N_G(y_{i_1})$ for each $j \in [2, \alpha']$.
  Thus, we have~$\Delta_{G[X', Y']} < c$.
  Note that we choose $x_i$ and $y_i$ such that there is an edge~$x_i y_i \in E(G)$ for each $i \in [\alpha]$.
  So $G[X', Y']$ has no isolated vertices.
  Therefore, it follows from \Cref{lemma:im:ramseyd2} that there is an induced matching $\{ x_{i_1'} y_{i_1'}, \dots, x_{i_k'} y_{i_k'} \}$ of size $k$ in $G[X', Y']$.
  Now, the set~$\{ x_{i_1'}, \dots, x_{i_k'} \}$ is an irredundant set in $G$, where $y_{i_j'}$ is a private neighbor of $x_{i_j'}$ for each~$j \in [k]$.
\end{proof}

\section{Conclusion}
We have demonstrated that the $c$-closure of a graph can be exploited in the design of
parameterized algorithms for well-studied graph problems. We believe that the $c$-closure
could become a standard secondary parameter just as the maximum degree~$\Delta$ or the
degeneracy~$d$ of the input graph and that studying problems with respect to this
parameter may often lead to useful tractability results, as evidenced by the FPT-algorithms and kernelizations that have been obtained recently for $c$-closed graphs~\cite{FRSWW20,HR20,KMR+22,KKNS22,KKS20a,KKS21,KN21}. In essence, whenever one obtains
a fixed-parameter algorithm that uses~$\Delta$ as one of its
parameters, one should ask whether~$\Delta$ can be replaced by the~$c$-closure of the
input graph. 
As concrete applications of the $c$-closure parameterization, one could
consider further graph problems that are hard with respect to the solution size.
In the extended abstract of this work we asked whether \textsc{Perfect Code} is fixed-parameter tractable with respect to~$c+k$ where~$k$ is the size of the code~\cite{KKS20}.
Very recently, Kanesh et al.~\cite{KMR+22} answered this question positively.
For this, Kanesh et al. exploited, among other things, our Ramsey bound on $c$-closed graphs (see \Cref{lemma:ramsey}).
Further problems to investigate could be \textsc{$r$-Regular Induced
  Subgraph} which is W[1]-hard when parameterized by the subgraph size~\cite{MT09} or
cardinality-constrained optimization problems in graphs where we search for a vertex set of size exactly~$k$ maximizing some objective function~\cite{Cai08}.
 These
problems are often fixed-parameter tractable for the combination of the cardinality
constraint~$k$ and the maximum degree~$\Delta$~\cite{Cai08,KS15}. We showed recently that a certain class of such problems which is concerned with maximizing or minimizing the number of incident edges is fixed-parameter tractable with respect to the combination of~$k$ and~$c$~\cite{KKNS22}.  Are there further fixed-cardinality optimization problems that are also
fixed-parameter tractable for this combined parameter? 

\bibliographystyle{plainurl}

\end{document}